\documentclass[manuscript,screen,acmsmall]{acmart}

\AtBeginDocument{%
  \providecommand\BibTeX{{%
    \normalfont B\kern-0.5em{\scshape i\kern-0.25em b}\kern-0.8em\TeX}}}

\setcopyright{acmcopyright}
\copyrightyear{2021}
\acmYear{2021}
\acmDOI{}

\usepackage[utf8]{inputenc}

\usepackage{amssymb}
\usepackage{amsmath}
\usepackage{stmaryrd}
\usepackage{listings}
\usepackage{proof}
\usepackage[all]{xypic}
\usepackage{ebproof}
\usepackage{xspace}
\usepackage{marvosym}
\usepackage{microtype}
\usepackage{cleveref}
\usepackage{ifthen}
\usepackage{xparse}
\usepackage{mathrsfs}
\usepackage{tikz-cd}
\usepackage{fontawesome}

\allowdisplaybreaks

\crefname{figure}{fig.}{fig.}
\Crefname{figure}{Fig.}{Fig.}

\renewcommand{\hat}[1]{\widehat{#1}}

\newcommand{\LF}{{LF}\xspace}
\newcommand{\id}{\mathsf{id}}

\newcommand{\comma}{,\,}
\newcommand{\semi}{;\,}
\newcommand{\of}{{:}\,}
\newcommand{\type}{\ \text{type}}

\newcommand{\SeqTm}[3]{#1 \vdash #2 \colon #3}

\newcommand{\kw}[1]{\texttt{#1}}
\newcommand{\sem}[1]{\llbracket #1 \rrbracket}

\newcommand{\CC}{\mathbb{D}}

\newcommand{\op}[1]{{#1}^{\mathrm{op}}}
\newcommand{\psh}[1]{\mathrm{Set}^{\op \CC}}

\newcommand{\Obj}{\kw{Obj}}
\newcommand{\ctx}{\kw{ctx}}
\newcommand{\cvar}{u}
\newcommand{\tm}{\kw{tm}}
\DeclareDocumentCommand{\trm}{o}{
  \IfNoValueTF{#1}{
    \kw{trm}
  } {
    \kw{trm}\;#1
  }
}
\newcommand{\ty}{\kw{ty}}

\newcommand{\tty}{\textsf{ty}}
\DeclareDocumentCommand{\ttrm}{o}{
  \IfNoValueTF{#1}{
    \textsf{trm}
  } {
    \textsf{trm}[#1]
  }
}

\newcommand{\Ctx}{\mathsf{Ctx}}
\newcommand{\Ty}[1]{\mathsf{Ty}(#1)}
\newcommand{\hTy}[1]{\hat{\mathsf{Ty}}(#1)}
\newcommand{\tTy}[1]{\mathtt{Ty}(#1)}
\newcommand{\CTm}{\mathsf{Tm}}
\newcommand{\Tm}[2]{\CTm(#1, #2)}
\newcommand{\tTm}[2]{\mathtt{Tm}(#1, #2)}

\newcommand{\Tmv}[2]{\CTm_v(#1, #2)}
\newcommand{\App}{\mathsf{App}}
\newcommand{\CEl}{\kw{El}}
\newcommand{\El}[1]{\CEl\,#1}

\newcommand{\CT}[2]{#1 \vdash #2}
\newcommand{\CTV}[2]{#1 \vdash_{v} #2}
\newcommand{\mpi}{\Pi}

\newcommand{\Type}{\texttt{Type}}

\newcommand{\tletbox}[3]{\kw{let}~\kw{box}~ #1 = #2\ \kw{in}\ #3}
\newcommand{\tlet}[3]{\kw{let}\  #1 = #2\ \kw{in}\ #3}

\newcommand{\tunbox}[1]{\kw{unbox}~#1}

\newcommand{\y}{\mathrm{y}}
\newcommand{\iso}{\cong}

\renewcommand{\Box}{\flat}

\newdir{ >}{{}*!/-10pt/@{>}}

\usepackage{letltxmacro}
\newcommand*{\SavedLstInline}{}
\LetLtxMacro\SavedLstInline\lstinline
\DeclareRobustCommand*{\lstinline}{%
  \ifmmode
    \let\SavedBGroup\bgroup
    \def\bgroup{%
      \let\bgroup\SavedBGroup
      \hbox\bgroup
    }%
  \fi
  \SavedLstInline
}

\DeclareMathAlphabet{\mathpzc}{OT1}{pzc}{m}{it}

\DeclareDocumentCommand{\judge}{ o m } {
  \IfNoValueTF {#1}
  {\Gamma \vdash #2}
  {#1 \vdash #2}
}

\DeclareDocumentCommand{\typing}{ o m m } {
  \judge[#1]{#2 : #3}
}

\DeclareDocumentCommand{\mjudge}{ o o m } {
  \IfNoValueTF {#1}
  {
    \IfNoValueTF {#2}
    {\judge[\Gamma ; \Psi]{#3}}
    {\judge[\Gamma ; #2]{#3}}
  }
  {
    \IfNoValueTF {#2}
    {\judge[#1 ; \Psi]{#3}}
    {\judge[#1 ; #2]{#3}}
  }
}

\DeclareDocumentCommand{\mtyping}{ o o m m } {
  \mjudge[#1][#2]{#3 : #4}
}

\DeclareDocumentCommand{\istype}{o o m} {
  \mjudge[#1][#2]{#3\texttt{ type}}
}

\DeclareDocumentCommand{\mmjudge}{ o o m } {
  \IfNoValueTF {#1}
  {
    \IfNoValueTF {#2}
    {\judge[\Gamma \;|\; \Theta]{#3}}
    {\judge[\Gamma \;|\; #2]{#3}}
  }
  {
    \IfNoValueTF {#2}
    {\judge[#1 \;|\; \Theta]{#3}}
    {\judge[#1 \;|\; #2]{#3}}
  }
}

\DeclareDocumentCommand{\mmtyping}{ o o m m } {
  \mmjudge[#1][#2]{#3 : #4}
}

\DeclareDocumentCommand{\mistype}{o o m} {
  \mmjudge[#1][#2]{#3\texttt{ type}}
}

\DeclareDocumentCommand{\isctx}{o m} {
  \judge[#1]{#2\texttt{ ctx}}
}

\DeclareDocumentCommand{\Istype}{o m} {
  \judge[#1]{#2\texttt{ type}}
}

\newcommand{\pipe}{\bnfalt}
\newcommand{\Set}{\ensuremath{\mathpzc{Set}}\xspace}

\newcommand{\boxit}[1]{\kw{box}~#1}
\newcommand{\letbox}[3]{\tletbox{#1}{#2}{#3}}
\newcommand{\lift}{\texttt{lift}}

\newcommand{\vjudge}[2]{\ensuremath{#1 \vdash_{v} #2}}
\newcommand{\quot}[1]{\ensuremath{\lceil #1 \rceil}}
\newcommand{\unquot}[1]{\ensuremath{\lfloor #1 \rfloor}}
\newcommand{\unquote}[2]{\ensuremath{\lfloor #1 \rfloor_{#2}}}
\newcommand{\deffun}[2]{\ensuremath{\texttt{fn}\;#1\Rightarrow#2}}
\newcommand{\induct}[4]{\ensuremath{\texttt{rec}^{#1}\;#2\;#3\;#4}}
\newcommand{\ctxext}[1]{\ensuremath{\langle #1 \rangle}}
\newcommand{\intp}[1]{\ensuremath{\llbracket #1 \rrbracket}}

\newcommand{\multicoltext}[1]{\multicolumn{1}{p{12cm}}{#1\newline}}

\usepackage{mathpartir}

\usepackage{todonotes}

\newcommand{\cocon}{\textsc{Cocon}\xspace}

\newcommand{\hatctx}[1]{\hat{#1}}
\newcommand\floor[1]{\lfloor#1\rfloor}
\newcommand\ceil[1]{\lceil#1\rceil}
\newcommand{\cbox}[1]{\ceil {#1}}
\newcommand{\unbox}[2]{\floor {#1}_{#2}}
\newcommand{\unboxc}[1]{#1}

\newcommand{\lfs}[2]{} 

\definecolor{DimGrey}{rgb}{0.8,0.8,0.8}

\lstdefinelanguage{ContextualML}
{
  morekeywords={and, block, case, of, mlam, fn, impossible, let, in, schema,
    some, rec, type, ctype, prop, stratified, inductive, coinductive, LF, if, then,
    else, total, with},
  keepspaces=true,
  sensitive,
  morecomment=[l]{\%},
  morecomment=[n]{\%\{}{\}\%},
  morestring=[b]"
}[keywords,comments,strings]

\newcommand{\teletype}[1]{\ensuremath{\mathtt{#1}}}
\newcommand{\ttapp}{\teletype{app}}
\newcommand{\ttlam}{\teletype{lam}}
\newcommand{\tlam}{\mathsf{lam}}
\newcommand{\tarr}{\mathsf{arr}}
\newcommand{\ttarr}{\teletype{arr}}

\newcommand{\tapp}{\mathsf{app}}

\newcommand{\ttobj}{\teletype{o}}
\newcommand{\tobj}{\mathsf{o}}

\newcommand{\IH}{{\mathcal{I}}}

\newcommand{\rappto}{~}

\newcommand{\ann}[1]{ \breve{#1}}

\newcommand{\clam}{\ensuremath{\mathtt{lam}}~\xspace}
\newcommand{\capp}{\ensuremath{\mathtt{app}}~\xspace}
\newcommand{\cletv}{\ensuremath{\mathtt{letv}}~\xspace}

\newcommand{\bnfas}{\;\mathrel{::=}\;}
\newcommand{\bnfalt}{\, \mid \,}
\newcommand{\wvec}[1]{\overrightarrow{#1}}

\newcommand{\wwk}[1]{\textbf{wk}}
\newcommand{\wk}[1]{\ensuremath{\mathsf{wk}_{#1}}}

\renewcommand{\arrow}{\Rightarrow}

\newcommand{\mlam}{\leftthreetimes^{\!\!\scriptscriptstyle\Box}}
\newcommand{\lamb}{\leftthreetimes}
\newcommand{\app}[2]{#1~#2}

\newcommand{\mto}{\mapsto}
\newcommand{\MBox}[1]{\mathsf{box}~#1}

\newcommand{\tmctx}{\ensuremath{\mathsf{ctx}}}
\newcommand{\R}{\mathcal{B}}

\newcommand{\titer}[3]{\mathsf{rec}^{#3}~\ensuremath{{#1}}}

\newcommand{\tmfn}[2]{\mathsf{fn}\; #1 \Rightarrow #2}

\newcommand{\tmrec}[4]{\mathsf{rec}^{#1} (#2 \mid #3 \mid #4)~}

\newcommand{\tightoverset}[2]{%
  \mathop{#2}\limits^{\vbox to -.5ex{\kern-0.95ex\hbox{$#1$}\vss}}}

\newcommand{\thelock}{\text{\faLock}}

\newcommand{\el}[1]{\scriptstyle\int\!#1}
\newcommand{\elop}[1]{\op{{\scriptstyle\int}#1}}

\begin{document}

\title{A Category Theoretic View of Contextual Types: from Simple Types to Dependent Types}

\author{Jason Z. S. Hu}
\email{zhong.s.hu@mail.mcgill.ca}
\affiliation{
  \institution{McGill University}
  \streetaddress{McConnell Engineering Bldg. 3480 University St.,}
  \city{Montr\'eal}
  \state{Qu\'ebec}
  \country{Canada}
  \postcode{H3A 0E9}
}

\author{Brigitte Pientka}
\email{bpientka@cs.mcgill.ca}
\affiliation{
  \institution{McGill University}
  \streetaddress{McConnell Engineering Bldg. 3480 University St.,}
  \city{Montr\'eal}
  \state{Qu\'ebec}
  \country{Canada}
  \postcode{H3A 0E9}
}

\author{Ulrich Schöpp}
\email{schoepp@fortiss.org}
\affiliation{%
  \institution{fortiss GmbH}
  \city{Munich}
  \country{Germany}
}

\renewcommand{\shortauthors}{Hu, et al.}

\begin{abstract}
We describe the categorical semantics for a simply typed
variant and a simplified dependently typed variant of \cocon, a contextual modal type theory where the box modality mediates
between the weak function space that is used to represent higher-order
abstract syntax (HOAS) trees and the strong function space that
describes (recursive)  computations about them. What makes \cocon
different from standard type theories is the presence of first-class
contexts and contextual objects to describe syntax trees that are closed
with respect to a given context of assumptions. Following M.~Hofmann's
work, we use a presheaf model to characterise HOAS
trees. Surprisingly, this model already provides the necessary
structure to also model \cocon. In particular, we can capture the
contextual objects of \cocon using a comonad~$\flat$ that restricts
presheaves to their closed elements.
This gives a simple semantic characterisation of the invariants of
contextual types (e.g.~substitution invariance) and
identifies \cocon as a type-theoretic syntax of presheaf models.
We further extend this characterisation to dependent types using categories with
families and show that we can model a fragment of \cocon without
recursor in the Fitch-style dependent modal type theory presented by Birkedal et. al..
\end{abstract}


\maketitle

\section{Introduction}
\label{sect:intro}
A fundamental question when defining, implementing, and working with languages and logics is:
How do we represent and analyse syntactic structures?
Higher-order abstract syntax \citep{Pfenning88pldi} (or lambda-tree
syntax \citep{Miller:LambdaTree99}) provides a deceptively
simple answer to this question. The basic idea to
represent syntactic structures is to map uniformly binding structures in
our object language (OL) to the function space in a meta-language thereby inheriting $\alpha$-renaming and capture-avoiding substitution.
In the logical framework \LF \citep{Harper93jacm}, for example, we can
define a small functional programming language
consisting of functions, function application, and let-expressions
using a type \lstinline!tm! as follows:

\begin{lstlisting}
lam : (tm -> tm) -> tm.        letv: tm -> (tm -> tm) -> tm.          app : tm -> tm -> tm.
\end{lstlisting}

The object-language term $(\mathsf{lam}\;x.\,\mathsf{lam}\;y.\,\mathsf{let}\;w =
x\;y\;\mathsf{in}\;w\;y)$ is then encoded as
\lstinline!lam \x.lam \y.letv (app x y) \w.app w y!
using the \LF abstractions to model binding.
Object-level substitution is modelled through \LF application; for instance, the fact that
$((\mathsf{lam}~x.M)~N)$ reduces to $[N/x]M$ in our object language is expressed as
\lstinline[basicstyle=\ttfamily\footnotesize]!(app (lam M) N)!
reducing to
\lstinline[basicstyle=\ttfamily\footnotesize]!(M N)!.

This approach is elegant and can offer substantial benefits: we can
treat objects equivalent modulo renaming and do not need to define
object-level substitution.

However, we not only want to just construct HOAS trees, but also to analyse
them and to select sub-trees. This is challenging, as sub-trees are
context sensitive. For example, the term
\lstinline[basicstyle=\ttfamily\footnotesize]!letv (app x y) \w.app w y!
only makes sense in a context
\lstinline[basicstyle=\ttfamily\footnotesize]!x:tm,y:tm!.
Moreover, one cannot simply extend \LF to allow syntax analysis.
If one simply added a recursion combinator to  \LF, then
it could be used to define many functions
\lstinline[basicstyle=\ttfamily\footnotesize]!M: tm -> tm!
for which
\lstinline[basicstyle=\ttfamily\footnotesize]!lam M!
would not represent an object-level syntax term~\citep{Hofmann:LICS99}.

Contextual types \citep{Nanevski:ICML05,Pientka:POPL08,gabbay_denotation_2013} offer a
type-theoretic solution to these problems by reifying the typing
judgement, i.e.~that \lstinline[basicstyle=\ttfamily\footnotesize]!letv (app x y) \w.app w y!
has type \lstinline[basicstyle=\ttfamily\footnotesize]!tm! in the
context \lstinline[basicstyle=\ttfamily\footnotesize]!x:tm,y:tm!,
as a \emph{contextual type} $\cbox{x{:}\tm,y{:}\tm \vdash \tm}$.
The contextual type $\cbox{x{:}\tm,y{:}\tm \vdash \tm}$ describes a set of terms of type $\tm$ that may contain variables $x$ and $y$. In particular, the contextual object
$\cbox{x, y \vdash \cletv (\capp x~y)~\lambda w.\capp w~y}$
has the given contextual type.
By abstracting over contexts and treating contexts as first-class, we can
now recursively analyse HOAS trees
\citep{Pientka:POPL08,Pientka:PPDP08,Pientka:TLCA15}.
Recently, \citet{Pientka:LICS19} further generalised these ideas and presented a contextual modal type
theory, \cocon, where we can mix HOAS trees and computations, i.e. we
can use (recursive) computations to analyse and traverse (contextual)
HOAS trees and we can embed computations within HOAS trees.
This line of work provides a syntactic perspective to the question
of how to represent and analyse syntactic structures with binders, as
it focuses on decidability of type checking and
normalisation. However, its semantics remains not
well-understood. What is the semantic meaning of a contextual type?
Can we semantically justify the given induction principles? What is
the semantics of a first-class context?

While a number of closely related categorical models of abstract syntax with bindings
\citep{Hofmann:LICS99,Fiore:LICS99,Gabbay:LICS99} were proposed around
2000, the relationship of these models to concrete type-theoretic
 languages for computing with HOAS structures was
tenuous.
In this paper, we give a category-theoretic semantics for
\cocon. This provides
semantic perspective of contextual types and first-class
contexts. Maybe surprisingly, the presheaf model introduced by \citet{Hofmann:LICS99}
already provides the necessary structure
to also model contextual modal type theory.
Besides the standard structure of this model, we only rely on two key 
concepts: a box (necessity) modality which is the same modality discussed by \citet[Section 6]{Hofmann:LICS99} and a cartesian closed universe of representables.
In the first half of this paper,
we focus on the special case of \cocon where the HOAS trees are simply-typed.
Concentrating on the simply-typed setting allows us to
introduce the main idea without the additional complexity that
type dependencies bring with them. 
The dependently-typed case is reserved in \Cref{sec:depty}, in which we study the
semantics using a categorical framework for dependent types, categories with
families (CwFs). We extend the model in the simply typed case such that the domain category
is a CwF, the structure of which is preserved by the Yoneda embedding. In
\Cref{sec:fitch}, we study \cocon's relation with a Fitch-style system and show
their similarity.

Our work provides a semantic foundation to \cocon and can serve as a
starting point to investigate connections to similar work.
First, our work connects \cocon to other work on internal languages for
presheaf categories with a $\flat$-modality, such as spatial
type theory~\citep{shulman_2018} or crisp type
theory~\citep{Licata:FSCD18}.
Second, it may help to understand the relations of \cocon to type theories
that use a modality for metaprogramming and intensional recursion, such as~\citep{DBLP:journals/corr/Kavvos17}.
While \cocon is built on the same general ideas,
a main difference seems to be that \cocon distinguishes between
HOAS trees and computations, even though it allows mixed use of them.
We hope to clarify the relation by providing a semantical perspective.

This paper is an extended version of the conference
paper~\cite{DBLP:conf/fossacs/PientkaS20}. 

\section{Presheaves for Higher-Order Abstract Syntax}
\label{sect:presheaves}

Our work begins with the presheaf models for HOAS of~\citet{Hofmann:LICS99,Fiore:LICS99}.
The key idea of those approaches is to integrate substitution-invariance in
the computational universe in a controlled way.
For the representation of abstract syntax, one wants to allow only substitution-invariant
constructions.  For example, \lstinline[basicstyle=\ttfamily\footnotesize]!lam M! represents
an object-level abstraction if and only if
\lstinline[basicstyle=\ttfamily\footnotesize]!M! is a function that uses its argument
in a substitution-invariant way.
For computation with abstract syntax, on the other hand, one wants to
allow non-substitution-invariant constructions too.
Presheaf categories allow one to choose the desired amount of substitution-invariance.

Let $\CC$ be a small category.
The presheaf category $\hat\CC$ is defined to be the category~$\psh\CC$.
Its objects are functors $F\colon \op\CC \to \mathrm{Set}$, which are also called \emph{presheaves}.
Such a functor~$F$ is given by a set $F(\Psi)$ for each
object $\Psi$ of $\CC$ together with a function
$F(\sigma) \colon F(\Phi) \to F(\Psi)$ for any object $\Phi$ and $\sigma\colon \Psi \to \Phi$ in
$\CC$, subject to the functor laws.
The intuition is that~$F$ defines sets of elements in various $\CC$-contexts,
together with a $\CC$-substitution action.
A morphism $f\colon F \to G$ is a natural transformation, which is a family of
functions $f_\Psi \colon F(\Psi) \to G(\Psi)$ for any $\Psi$.
This family of functions must be natural, i.e.~commute with substitution
$f_\Psi \circ F(\sigma) = G(\sigma) \circ f_\Phi$.

For the purposes of modelling higher-order abstract syntax, $\CC$ will typically be
the term model of some domain-level lambda-calculus.
By domain-level, we mean the calculus that serves as the meta-level for
object-language encodings.
It is the calculus that contains constants like
\lstinline[basicstyle=\ttfamily\footnotesize]!lam!
and
\lstinline[basicstyle=\ttfamily\footnotesize]!app!
from the Introduction.
We call it domain-level to avoid possible confusion between different meta-levels later.
For simplicity,
let us for now use a simply-typed lambda-calculus with functions and
products as the domain language.
It is sufficient to encode the example from the Introduction and allows us to explain the main idea underlying our approach.

The term model of the simply-typed domain-level lambda-calculus forms a cartesian closed
category~$\CC$.
The objects of~$\CC$ are contexts $x_1 \of A_1, \dots, x_n \of A_n$
of simple types.
We use $\Phi$ and $\Psi$ to range over such contexts.
A morphism from $x_1 \of A_1, \dots, x_n \of A_n$ to $x_1 \of B_1, \dots, x_m \of B_m$
is a tuple $(t_1,\dots, t_m)$ of terms
$\SeqTm{x_1\of A_1,\dots ,x_n\of A_n}{t_i}{B_i}$ for $i=1,\dots, m$.
A morphism of type $\Psi \to \Phi$ in $\CC$ thus amounts to a
(domain-level) substitution that
provides a (domain-level) term in context~$\Psi$ for each of the variables in~$\Phi$.
Terms are identified up to $\alpha\beta\eta$-equality.
One may achieve this by using a de Bruijn encoding, for example,
but the specific encoding is not important for this paper.
The terminal object is the empty context, which we denote by~$\top$,
and the product $\Phi \times \Psi$ is defined by context concatenation.
It is not hard to see that any object $x_1\of A_1, \dots, x_n \of A_n$
is isomorphic to an object that is given by a context with a single variable,
namely $x_1 \of (A_1 \times \dots \times A_n)$.
This is to say that contexts can be identified with product types.
In view of this isomorphism, we shall allow ourselves to
consider the objects of $\CC$ also as types and vice versa.
The category~$\CC$ is cartesian closed, the exponential of $\Phi$ and
$\Psi$ being given by the function type $\Phi \to \Psi$ (where the objects are
considered as types).

The presheaf category $\hat\CC$ is a computational universe that both
embeds the term model~$\CC$ and that can represent computations about it.
Note that we cannot just enrich~$\CC$ with terms for computations if we
want to use HOAS.
In a simply-typed lambda-calculus with just the constant terms
\lstinline[basicstyle=\ttfamily\footnotesize]!app: tm -> tm -> tm!
and
\lstinline[basicstyle=\ttfamily\footnotesize]!lam: (tm -> tm) -> tm!,
each term of type \lstinline[basicstyle=\ttfamily\footnotesize]!tm! represents
an object-level term.
This would not be the true anymore, if we were to allow computations
in the domain language,
since one could define
\lstinline[basicstyle=\ttfamily\footnotesize]!M!
to be something like
\lstinline[basicstyle=\ttfamily\footnotesize]!(\x. if x represents an object-level application then M1 else M2)!
for distinct
\lstinline[basicstyle=\ttfamily\footnotesize]!M1!
and
\lstinline[basicstyle=\ttfamily\footnotesize]!M2!. In this case,
\lstinline[basicstyle=\ttfamily\footnotesize]!lam M!
would not represent an object-level term anymore. If we want to preserve a
bijection between the object-level terms and their representations in
the domain-language, we cannot allow case-distinction over whether a
term represents an object-level an application.

The category~$\hat\CC$ unites syntax with computations by allowing one to enforce various degrees of
substitution-invariance.
By choosing objects with different substitution actions, one can
control the required amount of substitution-invariance.

In one extreme, a set~$S$ can be regarded as a constant presheaf. Define the constant presheaf $\Omega S$ by
$\Omega S(\Psi)=S$ and $\Omega S(\sigma)=\mathsf{id}$ for all~$\Psi$ and~$\sigma$.
Since the substitution action is trivial, a morphism $\Omega S \to \Omega T$ in~$\hat\CC$ amounts to 
just a function from set~$S$ to set~$T$.
Since $\Omega S$ is thus essentially just the set~$S$, we shall write just~$S$ both for the set~$S$ and the presheaf~$\Omega S$.

The Yoneda embedding represents the other extreme.
For any object~$\Phi$ of~$\CC$, the presheaf $\y(\Phi)\colon \op \CC \to \mathrm{Set}$
is defined by $\y(\Phi)(\Psi) = \CC(\Psi, \Phi)$, which is the set of morphisms
from~$\Psi$ to~$\Phi$ in~$\CC$.
The functor action is pre-composition.
The presheaf $\y(\Phi)$ should be understood as the type of all domain-level substitutions
with codomain~$\Phi$.
An important example is $y(\tm)$. In this case, $\y(\tm)(\Psi)$ is
the set of all morphisms of type $\Psi \to \tm$ in~$\CC$.
By the definition of $\CC$, these correspond to domain-level terms of type~$\tm$
in context~$\Psi$.
In this way, the presheaf $\y(\tm)$ represents the domain-level terms of type~$\tm$.

The Yoneda embedding does in fact embed~$\CC$ into~$\hat \CC$ fully and faithfully,
where the action on morphisms is post-composition.
This means that~$\y$ maps a morphism $\sigma \colon \Psi \to \Phi$ in $\CC$
to the natural transformation $\y(\sigma) \colon \y(\Psi) \to \y(\Phi)$ that is
defined by post-composing with~$\sigma$.
This definition makes $\y$ into a functor $\y\colon \CC \to \hat\CC$ that is
moreover full and faithful: its action on morphisms is a bijection from $\CC(\Psi, \Phi)$
to $\hat\CC(\y(\Psi), \y(\Phi))$ for any~$\Psi$ and~$\Phi$. 
This is because a natural transformation $f\colon \y(\Psi) \to \y(\Phi)$ is, by naturality,
uniquely determined by $f_\Psi(\id)$, where $\id \in \CC(\Psi,\Psi)=\y(\Psi)(\Psi)$, and
$f_\Psi(\id)$ is an element of~$y(\Phi)(\Psi)=\CC(\Psi, \Phi)$.

Since $\CC$ embeds into $\hat\CC$ fully and faithfully, the term model of the
domain language is available in $\hat\CC$.
Consider for example $\y(\tm)$.
Since $\y$ is full and faithful, the morphisms from~$\y(\tm)$ to~$\y(\tm)$ in~$\hat\CC$
are in one-to-one correspondence with the morphisms from~$\tm$ to~$\tm$ in~$\CC$.
These, in turn, are defined to be substitutions and correspond to
simply-typed (domain-level) lambda terms with one free variable.
This shows that substitution invariance cuts down the morphisms from $\y(\tm)$ to $\y(\tm)$
in $\hat\CC$ just as much as one would like for HOAS encodings.

But $\hat\CC$ contains not just a term model of the domain language.
It can also represent computations about the domain-level syntax and
computations that are not substitution-invariant.
For example, arbitrary functions on terms can be represented
as morphisms from the constant presheaf $\Omega(\y(\tm)(\top))$ to $\y(\tm)$.
Recall that~$\top$ is the empty context, so that
$\y(\tm)(\top)$ is the set $\CC(\top,\tm)$, by definition, which is isomorphic to
the set of closed domain-level terms of type~$\tm$.
The morphisms from $\Omega(\y(\tm)(\top))$ to $\y(\tm)$ in $\hat\CC$ correspond to arbitrary
functions from closed terms to closed terms, without any restriction of
substitution invariance.

The restriction to the constant presheaf of closed terms can be generalised to
arbitrary presheaves.
Define a functor $\flat\colon \hat\CC \to \hat\CC$
by letting $\flat F$ be the constant presheaf $\Omega(F(\top))$,
i.e.~$\flat F(\Psi) = F(\top)$ and $\flat F(\sigma) = \id$.
Thus, $\flat$ restricts any presheaf to the set of its closed elements.
The functor~$\flat$ defines a comonad where the counit $\varepsilon_F \colon \flat F \to F$
is the obvious inclusion and
the comultiplication $\nu_F \colon \flat F \to \flat \flat F$ is the identity.
The latter means that the comonad~$\flat$ is idempotent.

The category~$\hat\CC$ not only embeds~$\CC$ and allows control over
how much substitution invariance is wanted in various constructions, it
also is very rich in structure, not least because it also embeds $\Set$.
We will show that~$\hat\CC$ models contextual types and computations about them.
We will gradually introduce the structure of~$\hat\CC$ that we need to this end.
We begin here by noting that~$\hat\CC$ is cartesian closed, in order to introduce
the notation for this structure.

Finite products exist in~$\hat\CC$ and are constructed pointwise.
\begin{align*}
  \top(\Gamma) &= \{*\}
  &
    (X\times Y)(\Gamma)
  &= X(\Gamma)\times Y(\Gamma)
\end{align*}
The Yoneda embedding preserves finite products, i.e.~$\y(X\times Y)\iso \y(X) \times
\y(Y)$.

The category $\hat\CC$ has exponentials.
The exponential $(X\Rightarrow Y)$ can be calculated using the Yoneda lemma.
We recall that the Yoneda lemma states that $Z(\Gamma)$ is naturally
isomorphic to $\hat\CC(\y(\Gamma), Z)$.
With this, we have:
\[
  (X\Rightarrow Y)(\Gamma)
  \ \iso\ \hat\CC(\y(\Gamma), X \Rightarrow Y)
  \ \iso\ \hat\CC(\y(\Gamma) \times X , Y)
\]
Since~$\y$ preserves finite products, we have in particular
$(\y(A) \Rightarrow \y(B))(\Gamma)
\iso \hat\CC(\y(\Gamma) \times \y(A) , \y(B))
\iso \hat\CC(\y(\Gamma \times A) , \y(B))
\iso \CC(\Gamma \times A, B)$.
In the case where $A=B=\tm$, this shows that the exponential $\tm \Rightarrow \tm$
represents terms with an additional bound variable. Given exponentials in $\hat\CC$,
$y$ also preserves exponentials, i.e.~$\y(X \Rightarrow Y)\iso \y(X) \Rightarrow
\y(Y)$.

\section{Internal Language}
\label{sec:internal}

To explain how $\hat\CC$ models higher-order abstract syntax and contextual
types, we need to expose more of its structure.  Presenting it directly in terms
of functors and natural transformations is somewhat laborious and the technical
details may obscure the basic idea of our approach. We therefore use dependent
types as an internal language for working with~$\hat\CC$.

It is well-known that presheaf categories like $\hat\CC$ furnish a model of a dependent type theory
that supports dependent products, dependent sums and extensional identity types,
among others, see e.g.~\cite{Jacobs:TCS93}. In this section we outline this dependent type
theory and how it is related to $\hat\CC$. Since the constructions are largely standard, our aim
here is mainly to fix the notation for the later sections.

With the cartesian closed structure defined above, it is already possible to use the
simply-typed lambda-calculus for working with~$\hat\CC$.
Morphisms of type
$X_1\times \dots \times X_n \to Y$ in~$\hat\CC$ can be considered
as terms $x_1\of X_1,\dots,x_n\of X_n \vdash t \of Y$ of the simply-typed lambda
calculus. The cartesian closed structure of~$\hat\CC$ is enough to interpret abstraction
and application terms. With this correspondence, the simply-typed lambda calculus may
be used as a language for defining morphisms~$\hat\CC$.
Notice in particular that morphisms
$X_1\times \dots \times X_n \to Y_1 \times \dots \times Y_m$
correspond to a list of terms $(t_1,\dots, t_m)$ of type
$x_1\of X_1,\dots,x_n\of X_n \vdash t_i \of Y_i$
for all~$i$. Thus, morphisms $X_1\times \dots \times X_n \to Y_1 \times \dots \times Y_m$
can be considered as substitutions $(t_1/y_1,\dots, t_m/y_m)$ of
terms in context $x_1\of X_1,\dots,x_n\of X_n$ for the
variables in $y_1\of Y_1,\dots, y_m\of Y_m$.

The category~$\hat\CC$ has enough structure to extend this idea to a dependently-typed
lambda-calculus in a similar way. Let us first explain how contexts, types and
terms relate to the structure of~$\hat\CC$.
\begin{itemize}
\item Contexts:
  Typing contexts~$\Gamma$, $\Delta$ still correspond to objects of~$\hat\CC$.
  Also, the morphisms $\Gamma \to \Delta$ in~$\hat\CC$ still correspond
  to substitutions, just like in the simply-typed case.
\item Types:
  Due to type dependencies, the types are not simply objects anymore.
  For each context~$\Gamma$, the set of types in context~$\Gamma$ is specified by a set~$\hTy \Gamma$.
  The elements of~$\hTy \Gamma$ do not directly appear in~$\hat\CC$, but they induce morphisms in~$\hat\CC$.
  Each type $X \in \hTy \Gamma$ determines an object $\Gamma . X$ and
  a \emph{projection map} $p_X \colon \Gamma . X \to \Gamma$.
  The intention is that $\Gamma . X$ represents the context $\Gamma, x\of X$ and that
  the projection maps represent the weakening substitution $\Gamma, x\of X \to \Gamma$.
\item Terms:
  A term $\Gamma \vdash m \of X$ appear as morphisms
  $m\colon \Gamma \to \Gamma. X$ in $\hat\CC$ with the property $\id = p_X \circ m$.
  In essence it is a substitution for the variables in $\Gamma, x\of X$ that is the identity
  on all variables other than~$x$.
\end{itemize}


Concretely, we define the set $\hTy \Gamma$ to consist of presheaves
$\elop \Gamma \to \Set$, where $\el \Gamma$ is the \emph{category of elements}
of~$\Gamma$, which is defined as follows: the objects of $\el \Gamma$ are the set
$\{ (\Psi, g) \mid g \in \Gamma(\Psi) \}$ and the morphisms $\sigma$ between
$(\Psi, g) \to (\Phi, g')$ are induced from $\sigma : \Psi \to \Phi$ subject
to $g = \Gamma(\sigma, g')$. We can show that this construction does form a
category. Intuitively, this construction use morphisms in $\CC$ as coherence
conditions, which are respected in all other constructions in the presheaf model. 

For a presheaf $X : \elop \Gamma \to \Set$ in $\hTy\Gamma$, we define the 
presheaf $\Gamma . X$ in $\hat \CC$ by
$ (\Gamma . X)(\Phi) = \{ (g, x) \mid g \in \Gamma(\Phi),\ x\in X(\Phi, g) \}$ 
with the canonical action on morphisms. The projection map~$p_X$ is just the 
first projection.

This interpretation of types generalizes the simply-typed case.
Any object~$X$ of $\hat\CC$ can be seen as a simple type. It can be 
lifted to a presheaf in $\overline X : \elop \Gamma \to \Set$ in $\hTy{\Gamma}$ 
by letting $\overline X(\Phi, g) = X(\Phi)$. 
Note that with this choice, $\Gamma . \overline X$ is identical to $\Gamma \times X$,
as one would expect from the simply-typed case.

With type dependencies, one needs a substitution operation on types.  For any
$\sigma \colon \Delta \to \Gamma$, there is a type substitution function
$(-)\{\sigma\}$ from $\hTy\Gamma$ to $\hTy\Delta$. It maps $X : \hTy\Gamma$ to
\begin{align*}
  X\{\sigma\}(\Phi : \CC, d : \Delta(\Phi)) := X(\Phi, \sigma(\Phi, d))
\end{align*}

The definition of $\hTy\Gamma$ justifies the intuition of contexts as types of
tuples as in the simply-typed case.  Consider, for example, a context of the
form $(\top. X) . Y$ for some types $X\in \hTy\top$ and
$Y\in \hTy{\top. X}$. The presheaf $(\top. X) . Y$ is defined such
that the elements of $((\top. X) . Y)(\Phi)$ have the form
$((*, x), y)$ for some~$x$ and~$y$. $*$ is the only element of a uniquely chosen
singleton set in $\Set$.  One should think of them as representing
values of the variables of the context $x\of X, y \of Y$.

All presheaf categories also support dependent function
spaces~\citep{Hofmann:NI97}. For $X : \hTy \Gamma$ and $Y : \hTy{\Gamma.X}$, we write
$\hat\Pi(X, Y) : \hTy\Gamma$ for the dependent function space. This is defined to be
\begin{align*}
  \hat\Pi(X, Y)(\Phi : \CC, d : \Gamma(\Phi)) := \{ f
  &: (\Psi : \CC)(\delta : \Psi \to \Phi)(x : X(\Psi, \Gamma(\delta, d))) \to Y(\Psi, (\Gamma(\delta, d), x)) \\
  &\mid \forall \Psi, \delta, x, \Psi', \delta' : \Psi' \to \Psi.\,
    f(\Psi', \delta \circ \delta', X(\delta', x)) = Y(\delta', f(\Psi, \delta, x))\ \}
\end{align*}
$\hat\Pi(X, Y)$ is a set of functions the values of which are coherent under the
morphisms of the base category $\CC$. We can show that this definition of $\hat\Pi$
types respect type substitutions. 

This outlines how the structure of~$\hat\CC$ relates to dependent type theory,
see e.g.~\citep{Jacobs:TCS93} for more details. We will use this type theory
as a convenient internal language to work with the structure of~$\hat\CC$.
It is well-known that~$\hat\CC$ has enough structure to support
dependent sums and extensional identity types,
among others in its internal dependent type theory, see e.g.~\citep{Jacobs:TCS93}.
We do not need to spell out the details of their interpretation.

We use Agda notation for the types and terms of this internal dependent type theory.
We write $(x \of S) \to T$ for a dependent function type
and write $\lamb x\of S.m$ and $m\ n$ for the associated lambda-abstractions and applications.
As usual, we will sometimes also write $S\to T$ for $(x \of S) \to T$ if~$x$ does not appear in~$T$.
However, to make it easier to distinguish the function spaces at various levels, we will
write $(x \of S) \to T$ by default even when~$x$ does not appear in~$T$.
We use $\tlet x m n$ as an abbreviation for $\app{(\lamb x\of T.n)} m$, as usual.
For two terms $m\of T$ and $n\of T$, we write $m =_T n$ or just $m = n$ for the
associated identity type.

In the spirit of Martin-Löf type theory, we will identify the basic types and
terms needed for our constructions successively as they are needed. 
In the following sections, we will expose the
structure of $\hat\CC$ step by step until we have enough to interpret contextual
types.

While much of the structure of $\hat\CC$ can be captured by adding rules and constants to
standard Martin-Löf type theory, for the comonad~$\flat$ such a formulation
would not be very satisfactory.
The issues are discussed by \citet[p.7]{shulman_2018}, for example.
To obtain a more satisfactory syntax for the comonad, we refine the internal type theory
into a modal type theory in which~$\flat$ appears as a necessity modality.
This approach goes back to \citet{DBLP:conf/tlca/BentonBPH93,Barber96,Davies:ACM01}
and is also used by recent work of \citet{shulman_2018,Licata:FSCD18} and others on working with the
$\flat$-modality in type theory.

We summarise here the typing rules for the $\flat$-modality which we will
rely on. To control the modality, one uses two kinds of variables.
In addition to standard variables $x\of T$, one has a second kind of so-called
\emph{crisp} variables $x {::} T$.
Typing judgements have the form $\Delta \mid \Theta \vdash m \of T$,
where $\Delta$ collects the crisp variables and~$\Theta$ collects the
ordinary variables.
In essence, a crisp variable $x {::} T$ represents an assumption of the form
$x \of \flat T$.
The syntactic distinction is useful, since it leads to a type theory that
is well-behaved with respect to substitution~\citep{Davies:ACM01,shulman_2018}.

The typing rules are closely related to those in
modal type systems~\citep{Davies:ACM01,Nanevski:ICML05},
where $\Delta$ is the typing context for modal
(global) assumptions and $\Theta$ for (local) assumptions,
and type systems for linear logic~\citep{Barber96}, where
$\Delta$ is the typing context for non-linear assumptions and $\Theta$ for linear assumptions.
\begin{small}
  \begin{mathpar}
    \inferrule*
    { }
    {\Delta, u{::}T, \Delta' \mid \Theta \vdash u \of T}

    \inferrule*
    { }
    {\Delta \mid \Theta, x\of T, \Theta' \vdash x \of T}

    \inferrule*
    {\Delta \mid \cdot \vdash m : T}
    {\Delta \mid \Theta \vdash \MBox m : \Box T}

    \inferrule*
    {\Delta \mid \Theta \vdash m : \Box T \\ \Delta, x{::}T \mid \Theta \vdash n : S}
    {\Delta \mid \Theta \vdash \tletbox x m n : S}
  \end{mathpar}
\end{small}
Given any term $m\colon T$ which only depends on modal variable
context $\Delta$, we can form the term $\MBox m \colon \Box T$.
We have a let-term $\tletbox x m n$ that takes a term $m\colon \Box T$
and binds it to a variable $x{::} T$.
The rules maintain the invariant that the free variables in a type~$\Box T$  are all crisp variables from the crisp context
$\Delta$.

The model has other structures.
For example, the rules for dependent products are:
\begin{small}
\[
  \begin{prooftree}
 \Hypo{\Delta \mid \Theta, x\of T \vdash m \of S}
  \Infer1[]{
    \Delta \mid \Theta \vdash \lamb x{:}T.m : (x{:}T) \to S
}
\end{prooftree}
\quad
\begin{prooftree}
  \Hypo{\Delta \mid \Theta \vdash m \of (y{:}T) \to S \quad
        \Delta \mid \Theta \vdash n \of T}
  \Infer1[]{
    \Delta \mid \Theta \vdash \app m n \of [n/y]S
}
  \end{prooftree}
\]
\end{small}%
Though the dependent function space is supported by the model, in our interpretation,
we only use the simple function space. 
It is also convenient to have a crisp variant of abstractions and applications~\cite{Nanevski:ICML05}:
\begin{small}
\[
  \begin{prooftree}
 \Hypo{\Delta, u{::}T \mid \Theta \vdash m \of S}
  \Infer1[]{
    \Delta \mid \Theta \vdash \mlam u{::}T.m : (u{::}T) \to^\flat S
}
\end{prooftree}
\quad
\begin{prooftree}
  \Hypo{\Delta \mid \Theta \vdash m \of (u{::}T) \to^\flat S \quad
        \Delta \mid \cdot \vdash n \of T}
  \Infer1[]{
    \Delta \mid \Theta \vdash \app m n \of [n/u]S
}
  \end{prooftree}
\]
\end{small}%
The superscripts $\flat$ of $\lamb$ and $\to$ indicate that we are referring to the
crisp variant.
Notice that in the application rule, $n$ is necessarily closed, i.e. the local context of $n$ is empty.
These rules allow us to directly operate on crisp variables, and we will interpret
computation-level functions of \cocon to this crisp function space. 
Despite its convenience, the full effect of introducing crisp functions to a comonadic
modality type theory is still unclear and we leave its investigation to the future. In
this paper, however, we do not make use of its full strength, but just use the syntax as
a notation for the semantic interpretation. 

When~$\Delta$ is empty, we shall
write just $\Theta \vdash m \of T$ for $\Delta \mid \Theta \vdash m \of T$.

Let us outline the categorical interpretation of the modality rules.
First recall from above that $\flat$ is a comonad on~$\hat\CC$.
For any type $X \in \hTy \Delta$
we define the type $\flat X \in \hTy {\flat \Delta}$ by $(\flat X)(\Phi, d) = X(\top, d)$.
Notice that this is well-defined, since~$d \in (\flat\Delta)(\Psi)$ implies $d\in\Delta(\top)$ for all $\Psi$,
by definition of~$\flat$. 
Notice that the definition of $\flat X$ makes $\flat (\Delta.X)$ the same as $\flat \Delta. \flat X$.

Since $\Delta$ only contains crisp variables, types in it are represented as $\flat X$.
In general, when a type $T$ appears in a context $\Delta \mid \Theta$, it is a type in 
$\hTy(\flat \Delta. \Theta)$.

A variable lookup into the crisp context requires application of the counit $\epsilon$. Consider
\begin{align*}
  v : \flat \Delta \to \flat \Delta . \flat T
\end{align*}
which is the variable projection of the crisp context. We further need to unwrap $\flat T$
using $\epsilon : \flat T \to T$. Thus we need another natural transformation:
\begin{align*}
  h &: \flat \Delta . \flat T \to \flat \Delta . T \\
  h(\Psi) &= (id_{\flat \Delta}, \epsilon_\Psi)
\end{align*}
Thus we have $h \circ v : \flat \Delta \to \flat \Delta . T$. We can show that this
natural transformation is a section of the projection map, due to the first component in $h$ is an identity
morphism. General variable lookups in the crisp context can then be obtained by
weakening. 

To model $\kw{box}$, let us first consider a section natural transformation $\flat \Delta \to
\flat \Delta . \flat T$ given another section $m : \flat \Delta \to
\flat \Delta . T$. Though there might be a more general formulation, we take advantage
of the fact that $\flat$ is idempotent in our model, and we can immediately have
$\flat m$ to be the desired natural transformation, because $\flat \flat \Delta =
\flat \Delta$. The general $\kw{box}$ with $\Theta$ can be obtained by weakening.

The interpretation of $\tletbox x m n$ is relatively easier. We consider as an example
the
special case where $\Theta$ is empty. Given
$m : \flat \Delta \to \flat \Delta . \flat T$ and
$n : \flat \Delta . \flat T \to \flat \Delta. \flat T . S$, we can just apply $m$ as a
substitution to $n$: $n\{m\} : \flat \Delta \to \flat \Delta . S\{m\}$.

Thus we see that the syntax can be interpreted as categorical constructs. In the rest of the
paper, we use syntactic translations to simplify our presentation. 

\section{From Presheaves to Contextual Types}
\label{sect:simple}

Armed with the internal type theory, we can now explore the
structure of~$\hat\CC$.

\subsection{A Universe of Representables}
\label{sect:yoneda-universe}
For our purposes, the main feature of $\hat\CC$ is that it embeds~$\CC$ fully and
faithfully via the Yoneda embedding.
In the type theory for $\hat\CC$, we may capture this embedding by means of a
Tarski-style universe. Such a universe is defined by a type of codes~$\Obj$ for
the objects of $\CC$ together with a decoding function that maps these codes into
types of the type theory for~$\hat\CC$.

Let~$\Obj$ be the set of objects of~$\CC$.
Recall from above, that any set can be considered as a constant presheaf with the trivial substitution
action, and thus as a type in the internal type theory of~$\hat\CC$. 
The terms of this type~$\Obj$ represent objects of $\CC$.
The cartesian closed structure of~$\CC$ gives us terms $\kw{unit}$, $\kw{times}$, $\kw{arrow}$ for
the terminal object~$\top$, finite products~$\times$ and the exponential
(function type).
We also have a term for the domain-level type~$\tm$.
\begin{align*}
  &\vdash \Obj\type
    &
  &\vdash \tm \colon \Obj
  &
  &\vdash \kw{times} \colon (a \of \Obj) \to (b \of \Obj) \to \Obj
  \\
  &
    &
  &\vdash \kw{unit} \colon \Obj
  &
  &\vdash \kw{arrow} \colon (a\of\Obj) \to (b\of\Obj) \to \Obj  
\end{align*}
Subsequently, we sometimes talk about objects of~$\CC$  when we intend
to describe terms of type~$\Obj$ (and vice versa).

The morphisms of $\CC$ could similarly be encoded as a presheaf with many term constants,
but this is in fact not necessary.
Instead, we can use the Yoneda embedding to decode elements of~$\Obj$ into actual types.
To this end, we use the following:
\[
  x\of \Obj \vdash \El x \type
\]
The type $\El$ is almost direct syntax for the Yoneda embedding.
The interpretation of $\El$ in~$\hat\CC$, given in detail below, is such that, for any object~$A$ of~$\CC$, the type $\El A$ is
interpreted by the presheaf~$\y(A)$.
Such a presheaf is called \emph{representable}.
One can think of $\El A$ as the type of all
morphisms of type~$\Psi \to A$ in~$\CC$ for arbitrary~$\Psi$.
Recall from above that a morphism of type $\Psi \to A$ in~$\CC$ amounts to a domain-level term
of type~$A$ that may refer to variables in~$\Psi$.
In this sense, one should think of $\El A$ as a type of domain-level terms of type $A$,
both closed and open ones.

That~$\El A$ is interpreted by~$\y A$ means that all constructions on~$\El A$ in the internal type theory are guaranteed to be substitution invariant.
In particular, since the Yoneda embedding is full and faithful, recall \Cref{sect:presheaves},
the type of functions $(x\of\El A) \to \El B$ corresponds to the morphisms of type $A \to B$ in $\CC$.
Any closed term of type ${(x:\El A)} \to \El B$ corresponds to such a morphism $A \to B$ in $\CC$ and
vice versa.
This is because~$\El A$ and~$\El B$ correspond to $\y A$ and $\y B$ respectively and 
the naturality requirements in $\hat\CC$ enforce substitution-invariance,
as outlined in \Cref{sect:presheaves}.
The type $(x:\El A) \to \El B$ thus does not represent arbitrary
functions from terms of type~$A$ to terms of type~$B$, but only substitution-invariant
ones.
If a function of this type maps a domain-level variable $x\of A$ (encoded as an element of $\El A$) to some term $M\of B$ (encoded as an element of $\El B$), then it must map any other $N\of A$ to $[N/x]M$.

In more detail, the interpretation of $\Obj\in \hTy\top$ and $\El \in \hTy\Obj$ of the above types in the internal type theory of $\hat\CC$ is given by:
\begin{align*}
    \Obj(\Psi, *) &= \{\Phi \mid \text{$\Phi$ is an object of $\CC$}\}
    &
    \El(\Psi, \Phi) &= \CC(\Psi, \Phi)
    \\
    \Obj(\sigma) &: \Phi \mapsto \Phi
    &
    \El(\sigma) &: f \mapsto f \circ \sigma
\end{align*}
Notice in particular that $(\Obj . \El)(\Psi) = \{(\Phi, f) \mid \text{$\Phi$ is object of $\CC$}, f\in\CC(\Psi, \Phi)\}$.
If, for any object~$A$ in~$\CC$, we substitute along the corresponding constant function $A\colon \top \to \Obj$, then we obtain 
$(\top. \El\{A\})(\Phi) = \{(A, f) \mid f\in\CC(\Phi, A)\}$.
This presheaf is isomorphic to~$\y A$.

We note that, while type dependencies often make it difficult to spell out types directly in terms of the categorical structure of~$\hat\CC$,
type dependencies on constant presheaves like~$\Obj$ are relatively easy to work with. 
This is because~$\Obj$ is just a set, so that the naturality constraints
of $\hat\CC$ are vacuous for functions out of~$\Obj$.
Instead of working with the dependent type directly, we can just work with all its instances.
For example, a term of type $(a\of \Obj) \to (b\of \Obj) \to  (x\of\El a) \to \El b$
is uniquely determined by a family of terms $(x\of\El A) \to \El B$ indexed by
objects~$A$ and~$B$ in~$\CC$.
We have the following lemma, which states that functions out of~$\Obj$ have independent values for all arguments.
\begin{lemma}
\label{lemma:deps}
  In the internal type theory of\/~$\hat\CC$, a closed term $t\colon (a\of \Obj) \to X$
  is in one-to-one correspondence with a family of closed terms $(t_A)_{A\in \Obj}$ such that
  $t_A\colon X[A/a]$. In particular, there is no uniformity condition on this family, 
  i.e.~for different objects~$A$ and~$B$, the terms $t_A$ and $t_B$ may be arbitrary unrelated 
  terms of types $X[A/a]$ and $X[B/a]$.
\end{lemma}
Notice that such a lemma would not be true, e.g., with $\El A$ instead of $\Obj$.
We have seen above that functions of type $\El A \to \El B$ correspond to morphisms of type~$A\to B$ in~$\CC$.
By definition of~$\CC$, a morphism~$A\to B$ corresponds to a domain-level term $x\of A \vdash t \colon B$.
Such terms are not in one-to-one correspondence with families of closed terms of type~$B$ indexed by closed terms of
type~$A$.

To summarise this section, by considering the Yoneda embedding as a decoding function~$\CEl$ of a universe \'a la Tarski, 
we get access to $\CC$ in the internal type theory of $\hat\CC$.
Since the universe consists of the representable presheaves, we call it
the \emph{universe of representables}.

The following lemmas state that the embedding preserves terminal
object, binary products and the exponential.
\begin{lemma}
  \label{lemma:terminal}
  The internal type theory of\/ $\hat\CC$ has a term
  $\vdash \kw{terminal} \colon \El \kw{unit}$, such that $x = \kw{terminal}$
  holds for any $x\colon \El \kw{unit}$.
\end{lemma}

\begin{lemma}
  \label{lemma:products}
  The internal type theory of\/ $\hat\CC$ justifies the terms below,
  such that
  $\kw{fst}\ (\kw{pair}\ x\ y) = x$,
  $\kw{snd}\ (\kw{pair}\ x\ y) = y$,
  $z = \kw{pair}\ (\kw{fst}\ z)\ (\kw{snd}\ z)$
  for all~$x, y, z$.

  \smallskip
  \(
  \begin{aligned}
    c\of\Obj\comma d\of \Obj
    &\vdash
      \kw{fst} \colon (z : \El{(\kw{times}\ c\ d)}) \to \El c
    \\
    c\of\Obj\comma d\of \Obj
    &\vdash
      \kw{snd} \colon (z : \El{(\kw{times}\ c\ d)}) \to \El d
    \\
    c\of\Obj\comma d\of \Obj
    &\vdash
      \kw{pair} \colon (x:\El c) \to (y:\El d) \to \El{(\kw{times}\ c\ d)}
      \end{aligned}
      \)
\end{lemma}

\begin{lemma}
  \label{lemma:exponentials}
  The internal type theory of\/ $\hat\CC$ justifies the terms below
  such that
  $\kw{arrow-i}\ (\kw{arrow-e}\ f) = f$ and
  $\kw{arrow-e}\ (\kw{arrow-i}\ g) = g$ for all~$f, g$.

  \smallskip
  \(
  \begin{aligned}
    c\of\Obj\comma d\of \Obj
    &\vdash
      \kw{arrow-e} \colon (x\of\El{(\kw{arrow}\ c\ d)}) \to (y\of\El c) \to \El d
    \\
    c\of\Obj\comma d\of \Obj
    &\vdash
      \kw{arrow-i} \colon (y\of(\El c \to \El d)) \to \El{(\kw{arrow}\ c\ d)}
  \end{aligned}
  \)
\end{lemma}

\begin{proof}
  \Cref{lemma:terminal,lemma:products} are consequences of the preservation of limits of
  the Yoneda embedding.

  For \Cref{lemma:exponentials}, it suffices, by \Cref{lemma:deps}, to establish an isomorphism
  between $\El{(\kw{arrow}\ A\ B)}$ and $(y\of\El A) \to \El B$ for all objects~$A$ and~$B$ of~$\CC$.
  By definition of~$\El$, this amounts to preservation of exponentials by~$\y$.
  The goal follows because we have the following isomorphisms, which are natural in~$\Gamma$:
  \begin{align*}
  (\y A\Rightarrow \y B) (\Gamma)
  &\ \iso\ \hat\CC(\y\Gamma, \y A \Rightarrow \y B)
    \ \iso\ \hat\CC(\y\Gamma \times \y A, \y B)
  \\
  &\ \iso\ \hat\CC(\y(\Gamma \times A), \y B)
   \ \iso\ \CC(\Gamma \times A, B)
  \\
  &\ \iso\ \CC(\Gamma, A \Rightarrow B)
   \ \iso\ \hat\CC(\y\Gamma, \y(A \Rightarrow B))
   \ \iso\ \y(A \Rightarrow B)(\Gamma)
\end{align*}
\end{proof}

\subsection{Higher-Order Abstract Syntax}

The last lemma in the previous section states that $\El A \to \El B$ is isomorphic to
$\El (\kw{arrow}\ A\ B)$.
This is particularly useful to lift HOAS-encodings from $\CC$ to $\hat\CC$.
For instance, the domain-level term constant \lstinline`lam: (tm -> tm) -> tm`
gives rise to an element of $\El {(\kw{arrow}\ (\kw{arrow}\ \tm\ \tm)\ \tm)}$.
But this type is isomorphic to $(\El \tm \to \El \tm) \to \El \tm$, by the lemma.

This means that the higher-order abstract syntax constants lift to~$\hat\CC$:
\begin{align*}
  \kw{app} &\colon (m\of\El \tm) \to (n\of\El \tm) \to \El \tm
  &
  \kw{lam} &\colon (m\of(\El \tm \to \El \tm)) \to \El \tm
\end{align*}

Once one recognises $\El A$ as $\y(A)$, the adequacy of this higher-order abstract
syntax encoding lifts from~$\CC$ to $\hat\CC$ as in~\citet{Hofmann:LICS99}.
For example, an argument~$M$ to $\kw{lam}$ has type $\El \tm \to \El \tm$, which
is isomorphic to $\El (\kw{arrow}\ \tm\ \tm)$.
But this type represents (open) domain-level terms $t\colon \tm \to \tm$.
The term $\kw{lam}\ M\colon \El \tm$ then represents the domain-level term
$\kw{lam}\ t\colon \tm$, so it just lifts the domain-level.

\subsection{Closed Objects}
\label{sect:modality}

One should think of $\Box T$ as the type of `closed' elements of $T$.
In particular, $\Box (\El A)$ represents morphisms of type $\top \to A$ in~$\CC$,
recall the definition of $\flat$ from \Cref{sect:presheaves}
and that $\El A$ corresponds to $\y A$.
In the term model $\CC$, the morphisms $\top \to A$ correspond to closed domain-language terms of type~$A$.
Thus, while $\El A$ represents both open and closed domain-level terms,
$\Box(\El A)$ represents only the closed ones.

This applies also to the type $\El A \to \El B$.
We have  seen above that $\El A \to \El B$ is isomorphic to
$\El {(\kw{arrow}\ A\ B)}$ and may therefore be thought of as containing the
terms of type~$B$ with a distinguished variable of type~$A$.
But, these terms may contain other free domain language variables.
The type $\Box (\El A \to \El B)$, on the other hand, contains only terms
of type~$B$ that may contain (at most) one variable of type~$A$.

Restricting to closed objects with the modality is useful because it disables
substitution-invariance.
For example, the internal type theory for $\hat\CC$ justifies a function
$\kw{is-lam}\colon (x{:}\Box (\El \tm)) \to \kw{bool}$ that
returns $\kw{true}$ if and only if the argument represents an object language
lambda abstraction.
We shall define it in the next section.
Such a function cannot be defined with type $\El \tm \to \kw{bool}$, since it would
not be invariant under substitution.
Its argument ranges over terms that may be open; which particularly includes domain-level
variables.
The function would have to return $\kw{false}$ for them, since a domain-level
variable is not a lambda-abstraction.
But after substituting a lambda-abstraction for the variable, it would have to return
$\kw{true}$, so it could not be substitution-invariant.

We note that the type $\Obj$ consists only of closed elements and that
$\Obj$ and~$\Box \Obj$ happen to be definitionally equal types (an isomorphism
would suffice, but equality is more convenient).

\subsection{Contextual Objects}
\label{sect:contextual}

Using function types and the modality, it is now possible to work with contextual objects
that represent domain level terms in a certain context, much like in~\citet{Pientka:POPL08,Pientka:TLCA15}.
A contextual type $\cbox{\CT{\Psi}{A}}$ is a boxed function type of the form
$\Box (\El \Psi \to \El A)$.
It represents domain-level terms of type~$A$ with variables from~$\Psi$.
Here, we consider the domain-level context~$\Psi$ as a term that encodes it.
The interpretation will make this precise.

For example, domain-level terms with up to two free variables now appear as terms of type
\[
\Box (\El ((\kw{times}\ (\kw{times}\ \kw{unit}\ \tm)\ \tm) \to \El \tm ),
\]
as the following example illustrates.
\[
  \begin{array}{ll}
\mathsf{box} ~ (\lamb \cvar\of\El ((\kw{times}\ (\kw{times}\ \kw{unit}\ \tm)\ \tm).
  &\kw{let}\ x_1 = \kw{snd}\ (\kw{fst}\ \cvar) \ \kw{in}\\
  &\kw{let}\ x_2 = {\kw{snd}\ \cvar} \ \kw{in}  \\
  & \quad {\kw{app}\ (\kw{lam}\ (\lamb x\of\El \tm.\, \kw{app}\ x_1\ x))\ x_2~)}
  \end{array}
\]
Here, the variables~$x_1$ and~$x_2$ are bound at the meta level, i.e. the internal
language. As we will see in the next section, the example interprets the 
open domain-level term $ \kw{app}\ (\kw{lam}\ (\lambda x. \kw{app}\ x_1\
x))\ x_2$ with domain-level variables $x_1{:}\tm$ and $x_2{:}\tm$. 

This representation integrates substitution as usual.
For example,
given crisp variables $m{::}\El ({\kw{times}\ c\ \tm}) \to \tm$
and $n{::} \El c \to \tm$ for contextual terms,
the term $\MBox{(\lamb \cvar\of\El c.\, m\ (\kw{pair}\ \cvar\ (n\ \cvar)))}$ represents substitution of~$n$ for the last
variable in the context of~$m$.

For working with contextual objects, it is convenient to lift the
constants~$\kw{app}$ and~$\kw{lam}$ to contextual types.
\begin{align*}
c \of \Obj
  &\vdash
    \kw{app}' \colon
        \Box(\El c \to \El \tm)
    \to \Box(\El c  \to \El \tm)
    \to \Box{(\El c \to \tm)}
  \\
  c{:}\Obj
  &\vdash
    \kw{lam}' \colon \Box(\El (\kw{times}\ c\ \tm) \to  \El \tm)
    \to \Box{(\El c \to \El \tm)}
\end{align*}
These terms are defined by:
\[
  \begin{array}{lcl}
\kw{app}' & :=& \lamb m, n.\,
\kw{let}~\kw{box}~{m'} = {m}\ \kw{in}\ \kw{let}~\kw{box}~{n'} = {n}\  \kw{in}\ \\
& & \quad\quad\quad~~\MBox{(\lamb \cvar\of\El c.\, \kw{app}\ (m'\ \cvar)\ (n'\ \cvar))}
\\
\kw{lam}' & := & \lamb m.\, \tletbox{m'}{m}{\MBox{(\lamb \cvar\of\El c.\, \kw{lam}\ (\lamb x\of\El \tm.\,\ m'\ (\kw{pair}\ \cvar\ x)))}}
  \end{array}
\]

A contextual type for domain-level variables (as opposed to arbitrary terms) can be
defined by restricting the function space in $\Box(\El \Psi \to \El A)$
to consist only of projections.
Projections are functions of the form $\kw{snd} \circ \kw{fst}^k$,
where we write~$\kw{fst}^k$ for the $k$-fold iteration $\kw{fst} \circ \dots \circ \kw{fst}$.
Let us write $\El \Psi \to_v \El A$ for the subtype of $\El \Psi \to \El A$ consisting only of projections.
The contextual type $\Box(\El \Psi \to_v \El A)$ is then a subtype of
$\Box(\El \Psi \to \El A)$.

With these definitions, we can express a primitive recursion scheme for contextual types. We write it in its general form where the result type $A$ can possibly depend on $x$. This is only relevant for the dependently typed case; in the simply typed case, the only dependency is on~$c$.
\begin{lemma}\label{lem:rec}
  Let $c\of\Obj\comma x\of\Box(\El c \to \El \tm) \vdash A\ c\ x \type$ and define:

  \smallskip
  \(
  \begin{small}
  \begin{aligned}
    X_{\kw{var}} &:= (c \of \Obj) \to (x\of \Box(\El c \to_v \El \tm)) \to  A\ {c}\ x
    \\
    X_{\kw{app}} &:= (c \of \Obj) \to (x, y \of \Box(\El c\to \El \tm)) \to\, A\ c \ x \to A \ {c} \ y \to A\ c \ (\app {\app {\kw{app}'}{x}}\ {y})
    \\
    X_{\kw{lam}} &:= (c \of \Obj) \to (x \of \Box(\El ({\kw{times}\ c\  \tm}) \to \El \tm)) \to A \ (\kw{times}\ c \ \tm) \ x
\to A\ {c}\ (\app{\kw{lam}'} {x})
  \end{aligned}
\end{small}
\)

\smallskip\noindent
Then, $\hat\CC$ justifies a term

\smallskip
  \(
  \vdash
  \kw{rec} \colon X_{\kw{var}} \to X_{\kw{app}} \to X_{\kw{lam}} \to (c \of \Obj) \to (x \of \Box(\El c \to \El \tm)) \to\, \app{A}{c}\ x
  \)

  \smallskip\noindent
  such that the following equations are valid.

  \smallskip
\(
  \begin{array}{lcl}
    \app{\kw{rec}\ t_{\kw{var}}\ t_{\kw{app}} \ t_{\kw{lam}}}{c}\ x
    &= & \app{t_{\kw{var}}}{c}\ x \text{\hspace{2cm} \text{if $x \of \Box(\El c \to_v \El \tm)$ }}
    \\
    \app{\kw{rec}\ t_{\kw{var}}\ t_{\kw{app}} \ t_{\kw{lam}}}{c}\ (\app{\app{\kw{app}'}s} t)
    &= &\app{t_{\kw{app}}}{c}\ s\ t
    \\
    \app{\kw{rec}\ t_{\kw{var}}\ t_{\kw{app}} \ t_{\kw{lam}}}{c}\ (\app {\kw{lam}'}{s})
    &= &\app{t_{\kw{lam}}}{c}\ s
  \end{array}
\)
\end{lemma}
\begin{proof}[outline]
To outline the proof idea, note first that
a function of type $(c \of \Obj) \to (x\of \Box(\El c \to \El \tm)) \to\, A\ c\ x$ in $\hat\CC$,
corresponds to an inhabitant of $A\ \Phi\ t$ for each concrete object $\Phi$ of~$\CC$ and each
inhabitant $t\colon \Box(\El \Phi \to \El \tm)$.
This is because naturality constraints for boxed types are vacuous (and $\Obj = \Box\Obj$).
Next, note that inhabitants of $\Box(\El \Phi \to \El \tm)$ correspond to domain-level
terms of type~$\tm$ in context $\Phi$ up to $\alpha\beta\eta$-equality.
We can perform a case-distinction on whether it is a variable, abstraction or application
and depending on the result use $t_{\kw{var}}$, $t_{\kw{app}}$ or $t_{\kw{lam}}$
to define the required inhabitant of $A\ \Phi\ t$.
\end{proof}

As a simple example for $\kw{rec}$, we can
define the function $\kw{is-lam}$ discussed above by\newline
$\kw{rec}\
(\lamb c, x.\, \kw{false})\
(\lamb c, x,y,r_x, r_y.\, \kw{false})\
(\lamb c, x, r_x.\, \kw{true})$.

\section{Simple Contextual Modal Type Theory}
\label{sect:sctt}

We have outlined informally how the internal dependent type theory of~$\hat\CC$ can model
contextual types.
In this section, we make this precise by giving the interpretation of
\cocon \citep{Pientka:LICS19}, a contextual modal type theory where we can work with
contextual HOAS trees and computations about them, into $\hat\CC$.
We will focus here on a simply-typed version of \cocon where we use a
simply-typed domain-language with constants~$\kw{app}$ and~$\kw{lam}$
and also only allow computations about
HOAS trees, but do not consider, for example, universes.
Concentrating on a stripped down, simply-typed version of \cocon
allows us to focus on the essential aspects, namely how to interpret
domain-level contexts and domain-level contextual objects and types semantically.
The generalisation to a dependently typed domain-level such as \LF in
\Cref{sec:depty} will be conceptually straightforward, although
more technical. Handling universes is an orthogonal issue.

\begin{figure}[hbt]
\begin{center}
  \begin{small}
\[
\begin{array}{p{4.8cm}@{~}l@{~}r@{~}l}
Domain-level types           & A, B        & \bnfas & \tm \bnfalt A \to B\\
Domain-level terms           & M, N       & \bnfas & \lambda x.M \mid M\,N \mid
 x \mid \tlam\bnfalt\tapp \mid \unbox t \sigma \\
Domain-level contexts        & \Psi, \Phi & \bnfas & \cdot \bnfalt \psi \bnfalt \Psi, x{:}A\\
Domain-level context (erased) & \hatctx{\Psi},\hatctx{\Phi} & \bnfas & \cdot \bnfalt \psi \bnfalt\hatctx{\Psi}, x\\
Domain-level substitutions   & \sigma   & \bnfas & \cdot \bnfalt \wk{\hatctx\Psi} \bnfalt \sigma, M
\\[0.25em]
\hline
\\[-0.75em]
Contextual types & T & \bnfas &
                     \Psi \vdash A \bnfalt
                     \CTV \Psi A
\\
Contextual objects & C & \bnfas &
                      \hatctx{\Psi} \vdash M
\\[0.25em]
\hline
\\[-0.75em]
Domain of discourse & \ann\tau & \bnfas & \tau \bnfalt \ctx \\
Types and Terms & \tau, \IH & \bnfas &  \cbox T
                             \bnfalt (y :\ann{\tau}_1) \arrow \tau_2
\\
 & t, s &  \bnfas &  y \bnfalt  \cbox C \bnfalt \titer{\vec\R}{}{\IH} \rappto\Psi~t
             \bnfalt \tmfn y t \bnfalt  t_1~t_2

\\
Branches  & {\R} & \bnfas & \Gamma \mto t
\\
Contexts & \Gamma & \bnfas & \cdot \bnfalt \Gamma, y:\ann\tau
\end{array}
\]
  \end{small}
\end{center}
\vspace{-0.15cm}
  \caption{Syntax of \cocon with a fixed simply-typed domain $\tm$}
  \label{fig:grammar}
\end{figure}

We first define our simply-typed domain-level with the type $\tm$
and the term constants $\tlam$ and $\tapp$ (see
\Cref{fig:grammar}). Following \cocon, we allow computations to be
embedded into domain-level terms via unboxing. The intuition is that
if a program $t$ promises to compute a value of type
$\cbox{x{:}\tm,y{:}\tm \vdash \tm}$, then we can embed $t$ directly
 into a domain-level object writing $\clam \lambda x. \clam \lambda y. \capp
 \unbox{t}{}~x$, unboxing $t$.
Domain-level objects (resp.~types) can be packaged together with their domain-level
context to form a contextual object (resp.~type).  Domain-level
contexts are formed as usual, but may contain context variables to
describe a yet unknown prefix.  Last, we include domain-level
substitutions that allow us to move between domain-level contexts.
The compound substitution $\sigma, M$ extends the substitution
$\sigma$ with domain $\hat\Psi$ to a substitution with domain
$\hat\Psi,x$, where $M$ replaces $x$. Following
\citet{Nanevski:ICML05,Pientka:LICS19}, we do not store the domain
(like $\hat\Psi$) in the substitution, it can always be recovered
before applying the substitution. We also include \emph{weakening
  substitution}, written as $\wk{\hatctx\Psi}$, to describe the
weakening of the domain $\Psi$ to $\Psi, \wvec{x{:}A}$. Weakening
substitutions are necessary, as
they allow us to express the weakening of a context variable $\psi$.
Identity is a special form of the $\wk{\hatctx\Psi}$ substitution,
which follows immediately from the typing rule of $\wk{\hatctx\Psi}$.
Composition is admissible.

We summarise the typing rules for domain-level terms and types in
\Cref{fig:lftyping}. We also include typing rules for domain-level
contexts. Note that since we restrict ourselves to a simply-typed
domain-level, we simply check that $A$ is a well-formed type.
We 
remark that equality for domain-level terms and substitution is modulo~$\beta\eta$.
In particular, $\unbox{\cbox{\hatctx{\Phi} \vdash
    N}}{\sigma}$ reduces to $[\sigma] N $.

\begin{figure}[htb]
  \centering\small
\[
\begin{array}{c}
\multicolumn{1}{p{12cm}}{
\fbox{$\Gamma ; \Psi \vdash M : A$}~~Term $M$ has type $A$ in domain-level context $\Psi$ and context $\Gamma$\newline}
\\
\infer{\Gamma ; \Psi \vdash x : A}{
\Gamma \vdash \Psi : \ctx&
x{:}A \in \Psi}
\qquad
\infer{\Gamma ; \Psi \vdash \tlam : (\tm \to \tm) \to \tm}
      {\Gamma\vdash \Psi : \ctx }
\qquad
\infer{\Gamma ; \Psi \vdash \tapp : \tm \to \tm \to \tm}
      {\Gamma\vdash \Psi : \ctx }
\\[0.5em]
\infer{\Gamma ; \Psi \vdash M~N : B}
      {\Gamma ; \Psi \vdash M : A \to B &
       \Gamma ; \Psi \vdash N : A}
\qquad
 \infer{\Gamma ; \Psi \vdash \lambda x.M : A \to B}
         {\Gamma ; \Psi, x{:}A \vdash M : B}
\qquad
\infer{\Gamma ; \Psi \vdash \unbox t {\sigma} : A}
                                          {\Gamma \vdash t : \cbox{\Phi \vdash A} \text{ or }
                                          \Gamma \vdash t : \cbox{\CTV \Phi A}
                         &
          \Gamma; \Psi \vdash \sigma : \Phi}
\\[1em]
\multicolumn{1}{p{12cm}}{
\fbox{$\Gamma ; \Phi \vdash \sigma : \Psi$}~~Substitution $\sigma$
  provides a mapping from the  (domain) context $\Psi$ to $\Phi$\newline}
\\[0.3em]
\infer{\Gamma ;\Psi, \wvec{x{:}A} \vdash \wk{\hatctx\Psi} : \Psi}
{\Gamma\vdash \Psi, \wvec{x{:}A} : \ctx }
\quad
 \infer{\Gamma ; \Phi \vdash \cdot : \cdot}{\Gamma \vdash \Phi :\ctx }
\quad
\infer{\Gamma ; \Phi \vdash \sigma, M : \Psi, x{:}A}
      {\Gamma ; \Phi \vdash \sigma : \Psi &
       \Gamma ; \Phi \vdash M : \lfs \sigma \Psi A}
\\[1em]
\multicolumn{1}{l}{
\fbox{$\Gamma \vdash \Psi : \ctx$}~~\mbox{Domain-level context $\Psi$ is well-formed}} \\[0.3em]
\infer{\Gamma \vdash \cdot : \ctx}{}
\quad
\infer{\Gamma \vdash \unboxc{y} : \ctx}{\Gamma(y) = \tmctx }
\quad
\infer{\Gamma \vdash \Psi, x{:}A : \ctx}
{\Gamma \vdash \Psi : \ctx
}
\end{array}
\]
\caption{Typing Rules for Domain-level Terms, Substitutions, Contexts}\label{fig:lftyping}
\end{figure}

In our grammar, we distinguish
between the contextual type $\Psi \vdash A$ and the more restricted
contextual type $\CTV \Phi A$ which characterises only variables of
type $A$ from the domain-level context $\Phi$. We give here two sample
typing rules for $\CTV \Phi A$ which are the ones used most in practice to
illustrate the main idea.
We embed contextual objects into computations via the
modality. Computation-level types include boxed contextual types, $\cbox{\CT \Phi A}$,  and
function types, written as $ (y:\ann{\tau}_1) \arrow
\tau_2$. We overload the function space and allow as domain
of discourse both computation-level types and the schema $\tmctx$ of
domain-level contexts, although only in the latter case $y$ can occur
in $\tau_2$. We use $\tmfn y t$ to introduce functions of both kinds. We also
overload function application $t\;s$ to eliminate function types $(y :
\tau_1) \arrow \tau_2$ and $(y : \tmctx) \arrow \tau_2$, although in
the latter case $s$ stands for a domain-level context.
We separate domain-level contexts
from contextual objects, as we do not allow functions that return a
domain-level context.

The recursor is written
as $\titer{\vec\R}{}{\IH} \rappto\Psi~t$. Here, $t$ describes a term of
type $\cbox{\Psi \vdash \tm}$ that we recurse over and
$\vec\R$ describes
the different branches that we can take depending on the value
computed by $t$. As is common when we have dependencies, we annotate the recursor with the typing invariant
$\IH$.
Here, we consider only the recursor over domain-level terms of type
$\tm$. Hence, we annotate it with $\IH = (\psi : \tmctx) \arrow
(y:\cbox{\psi \vdash \tm}) \arrow \tau$. To check that the recursor
$\titer{\R}{}\IH~\Psi~t$ has type $[\Psi/\psi]\tau$, we check that
each of the three branches has the specified type $\IH$. In the base
case, we may assume in addition to $\psi:{\tmctx}$ that we have a
variable $p:\cbox{\CTV{\unboxc{\psi}} \tm}$ and check that the body
has the appropriate type.
If we encounter a contextual object built with the domain-level constant $\tapp$, then we choose the branch $b_\tapp$. We assume $\psi \of \tmctx$, $m \of \cbox{\unboxc{\psi}\vdash \tm}$, $n \of \cbox{\unboxc{\psi}\vdash \tm}$, as well as $f_n$ and $f_m$ which stand for the recursive calls on $m$ and $n$ respectively. We then check that the body $t_\tapp$ is well-typed.
If we encounter a domain object built with the domain-level constant $\tlam$, then we choose the branch $b_\tlam$. We assume $\psi \of \tmctx$ and $m \of \cbox{\psi, x{:}\tm \vdash \tm}$ together with the recursive call $f_m$ on $m$ in the extended LF context $\psi, x{:}\tm$. We then check that the body  $t_\tlam$ is well-typed.
The typing rules for computations are given in
\Cref{fig:comptyping}. We omit the reduction rules here for brevity.

\begin{figure}[htb]
  \centering\small
  \[
    \begin{array}{c}
      \multicolumn{1}{p{12cm}}{
        \fbox{$\Gamma \vdash C : T$}~~ Contextual object $C$ has contextual type $T$\newline}
    \\
    \infer{\Gamma \vdash (\hatctx{\Psi} \vdash M) : (\Psi \vdash A)}{\Gamma; \Psi \vdash M : A}
    \qquad
    \infer{\Gamma \vdash (\hatctx{\Psi} \vdash x) : (\CTV \Psi A)}
      {\Gamma \vdash \Psi : \ctx&
       x{:}A \in \Psi}
\qquad
\infer{\Gamma \vdash (\hatctx{\Psi} \vdash \unbox{x}{\wk{\hatctx\Psi}})
                                  : (\CTV \Psi A)}
{ x{:}\cbox{\CTV \Phi A} \in \Gamma &
  \Gamma ; \Psi \vdash \wk{\hatctx\Psi} : \Phi}
    \end{array}
\]
\[
\begin{array}{c}
  \multicolumn{1}{p{12cm}}{
    \fbox{$\Gamma \vdash t : \tau$}~~ Term~$t$ has computation type~$\tau$}
  \\[1.5em]
  \infer{\Gamma \vdash y : \ann\tau}{y:\ann\tau \in \Gamma}
  \qquad
  \infer{\Gamma \vdash \cbox C : \cbox T}{\Gamma \vdash C : T}
  \qquad
\infer{\Gamma \vdash t~s : [s/y]\tau_2}
{\Gamma \vdash t : (y:\ann\tau_1) \arrow \tau_2 &
 \Gamma \vdash s : \ann\tau_1}
\qquad
\infer{\Gamma \vdash \tmfn y t : (y:\ann\tau_1) \arrow \tau_2}
        {\Gamma, y:\ann\tau_1 \vdash t : \tau_2 & \Gamma \vdash (y:\ann\tau_1) \arrow \tau_2: \type}
\\[0.75em]
\multicolumn{1}{l}{\mbox{Recursor over domain-level terms}~\IH = (\psi : \tmctx) \arrow (y:\cbox{\psi \vdash \tm}) \arrow \tau }\\[0.5em]
\infer
{\Gamma \vdash \tmrec {\IH} {b_v} {b_{\mathsf{app}}} {b_{\tlam}} \rappto \Psi~t : [{\Psi}/\psi]\tau}
{
\Gamma \vdash t :  \cbox{\Psi \vdash \tm}  & \Gamma \vdash \IH : \type  &
\Gamma \vdash b_v : \IH & \Gamma \vdash b_{\mathsf{app}} : \IH & \Gamma \vdash b_{\mathsf{lam}} : \IH
}
\\[1em]
\raisebox{2ex}{\mbox{Branch for Variable ($b_v$)}}\qquad
\inferrule*
{ \Gamma, \psi:\tmctx, p:\cbox{\CTV {\unboxc{\psi}} {\tm} }  \vdash  t_v : \tau}
{\Gamma \vdash ({\psi,p \mto t_v}) : \IH }~~~~~~~~~~~~~~~~~~~~~~~
 \\[0.5em]
\raisebox{2ex}{\mbox{Branch for Application $\tapp$ ($b_{\mathsf{app}}$)}}\qquad
\infer
{\Gamma \vdash (\psi, m, n, f_n, f_m \mto t_{\mathsf{app}}) : \IH}
{\Gamma, \psi:\tmctx, m{:}\cbox{\unboxc{\psi} \vdash \tm}, n{:}\cbox{\unboxc{\psi} \vdash \tm},
  f_m{:} \tau, f_n{:} \tau   \vdash   t_{\mathsf{app}} : \tau
}
\\[0.5em]
\raisebox{2ex}{\mbox{Branch for Function $\tlam$ ($b_{\tlam}$)}}\qquad
\infer{\Gamma \vdash \psi, m, f_m \mto t_{\tlam} : \IH}
{
\Gamma, \phi:\tmctx,    m{:}\cbox{\phi, x{:}\tm \vdash \tm},
         f_m{:}[ (\phi, x{:}\tm)/\psi ] \tau           \vdash  t_{\tlam} : [\phi/\psi]\tau
 }
\end{array}
\]
  \caption{Typing Rules for Contextual Objects and Computations}
  \label{fig:comptyping}
\end{figure}

We now give an interpretation of simply-typed \cocon in a presheaf
model with a cartesian closed universe of representables.
Let us first extend the internal dependent type theory
with the constant $\tm$ for modelling the
domain-level type constant $\tm$ and with the
constants
$\kw{app}\colon \El \tm \to \El \tm \to \El \tm$ and $\kw{lam} \colon
(\El \tm \to \El \tm) \to \El \tm$ to model the corresponding
domain-level constants $\tapp$ and $\tlam$.

\begin{figure}[htb]
  \centering\small
  \begin{alignat*}{4}
    &\mbox{Interpretation of domain-level types} \span\span\span \\[0.25em]
    &\sem{\tm}     &= &\tm \\
    &\sem{A \to B} &= &\kw{arrow}\ \sem{A}\ \sem{B}\\[0.5em]
    &\mbox{Interpretation of domain-level contexts} \span\span\span\\[0.25em]
    &\sem{\Gamma \vdash \psi : \ctx} & = & \psi \\[0.25em]
    &\sem{\Gamma \vdash \cdot : \ctx} & = & \kw{unit} \\[0.25em]
    &\sem{\Gamma \vdash (\Psi, x{:}A) : \ctx} & = & \kw{times}\ e \ \sem{A}
    \tag*{where $\sem{\Gamma \vdash \Psi : \ctx} = e$} \\[0.5em]
    &\mbox{Interpretation of domain-level terms, where } \sem{\isctx{\Psi}} = e\span\span\span
    \\[0.25em]
    &\sem{\Gamma ; \Psi \vdash x \colon A} & = &
    \lamb \cvar: \El e. \kw{snd}~(\kw{fst}^k~\cvar)
    \tag*{where $\Psi = \Psi_0, x{:}A, y_k{:}A_k, \ldots, y_1{:}A_1$} \\
    &\sem{\Gamma; \Psi \vdash \lambda x.\, M : A \to B} & = &
    \lamb \cvar: \El e. \kw{arrow-i}\ (\lamb x{:}\El \sem{A}.e'~(\kw{pair}~\cvar~x))
    \tag*{where $\sem{\Gamma ; \Psi, x{:}A \vdash M : B} = e'$} \\
    &\sem{\Gamma ; \Psi \vdash M\ N : B} & =&
    \lamb \cvar: \El e. \kw{arrow-e}\ (e_1~\cvar)\ (e_2~\cvar)
    \tag*{where $\sem{\Gamma ; \Psi \vdash M : A \to B} = e_1$} \\
    \tag*{ and $\sem{\Gamma ; \Psi \vdash N : A } = e_2$} \\
    &\sem{\Gamma ; \Psi \vdash \unbox t \sigma : A} & = &
    \tletbox{x}{e_1}{\lamb \cvar: \El e. x(e_2\; u)}
    \tag*{where $\sem{\Gamma \vdash t : \cbox{\Phi \vdash A}} = e_1$} \\
    \tag*{and $\sem{\Gamma ; \Psi \vdash \sigma : \Phi} = e_2$} \\
    &\sem{\Gamma ; \Psi \vdash \tapp :  \tm \to \tm \to \tm} &=&
    \lamb \cvar: \El e. \kw{arrow-i} (\lamb x{:}\El \tm.\, \kw{arrow-i}\ (\lamb y{:}\El \tm.\,
     \kw{app}\ x\ y)) \\
     &\sem{\Gamma ; \Psi \vdash \tlam : (\tm \to \tm) \to  \tm} &=&
     \lamb \cvar: \El e. \kw{arrow-i} (\lamb f{:}\El (\kw{arrow}~\tm~\tm).\, \kw{lam}\ (\lamb
     x{:}\El \tm.\, \kw{arrow-e}\ f\ x)) \\
    &\mbox{Interpretation of domain-level substitutions, where }  \sem{\isctx{\Psi}} = e\span\span\span \\
    &\sem{\Gamma ; \Psi \vdash \cdot : \cdot} & = & \lamb \cvar: \El e. \kw{terminal} \\
    &\sem{\Gamma ; \Psi \vdash (\sigma, M) : \Phi, x{:}A} & = &
    \lamb \cvar: \El e. \kw{pair}\ (e_1 \cvar)\ (e_2 \cvar)
    \tag*{where $\sem{\Gamma ; \Psi \vdash \sigma : \Phi} = e_1$ and $\sem{\Gamma ;
        \Psi \vdash M : A} = e_2$} \\
    & \sem{\Gamma ; \Psi, \wvec{x{:}A} \vdash \wk{\hatctx\Psi} : \Psi} &=&
    \lamb \cvar: \El e. \kw{fst}^n~\cvar
    \tag*{where $n = |\wvec{x{:}A}|$}
  \end{alignat*}
  \vspace{-1em}
  \caption{Interpretation of Domain-level Types and Terms}
  \label{fig:idom}
\end{figure}

We can now translate domain-level and computation-level types of \cocon
into the internal dependent type theory for $\hat\CC$. We do so by
interpreting the domain-level terms, types, substitutions, and
contexts (see Fig. \ref{fig:idom}). All translations are on well-typed
terms and types. Domain-level types are interpreted as the terms of type $\Obj$
in the internal dependent type theory that represent them.
Domain-level contexts are also interpreted as terms of type $\Obj$
 by $\sem{\Gamma \vdash \Psi : \ctx}$. For example, a domain-level context $x{:}\tm, y{:}\tm$ is
interpreted as  $\kw{times}~(\kw{times}~\kw{unit}~\tm)~\tm : \Obj$.
A domain-level substitution with domain~$\Psi$ and codomain $\Phi$ becomes a function
from $\El e'$ to $\El e$, where 
$e' = \sem{\Gamma \vdash \Psi : \ctx}$ and $e = \sem{\Gamma \vdash \Phi : \ctx}$. Thus we use a semantic function to
interpret a simultaneous substitution as usual.
As $e$ is some
product, for example $\kw{times}~(\kw{times}~\kw{unit}~\tm)~\tm$, the
domain-level substitution is translated into a function returning an n-ary tuple. A
weakening substitution
 $\Gamma ; \Psi, x{:}\tm \vdash \wk{\Psi} : \Psi$ is interpreted as
$\kw{fst}~\cvar$ where $\cvar \of \El (\kw{times}~e~\tm)$ and $e = \sem{\Gamma
  \vdash \Psi : \ctx}$. More generally, when we weaken a context
$\Psi$ by $n$ declarations, i.e. $\wvec{x{:}A}$, we interpret
$\wk{\Psi}$ as $\kw{fst}^n~\cvar$.
A well-typed domain-level term, $\Gamma ; \Psi \vdash M : A$, is mapped to a function
from $\cvar{:}\El \sem{\Gamma \vdash \Psi :  \ctx}$ to
$\El \sem{A}$.

Hence the translation of a well-typed domain-level term eventually introduces via a
$\lamb\ \cvar$ that stands for the term-level interpretation of a domain-level
context $\Phi$. Most cases in the interpretation just pass $\cvar$ along. The
exceptional case is $\Gamma ; \Phi \vdash \lambda x.M : A \to B$, because the body $M$
is translated into a function from an extended domain-level context $\Psi, x{:}A$. In
this case, we need to pair $\cvar$ with the variable $x$ obtained from the
constructor of the domain-level function. When we translate a variable $x$ where $\Phi = \Phi_0, x{:}A, y_k{:}A_k, \ldots, y_1{:}A_1$, we return $\kw{snd}~(\kw{fst}^k~\cvar)$.
We translate $\Gamma ; \Phi \vdash  \unbox t \sigma : A$ directly using $\kw{let
  box}$-construct. Intuitively, a substitution of terms is morphism composition. The
interpretation of the domain-level substitution $\sigma$ has given one morphism. The
computation term $t$ should give the other morphism. Since it has the contextual type
$\cbox{\Phi \vdash \tm}$ its translation will be of type $\flat(\El e \to \El \tm)$
where $e' = \sem{\Gamma \vdash \Phi : \ctx}$. We thus can obtain a $\El e \to \El \tm$
from \kw{let box}, in which we can compose $e_2$ to obtain the final morphism.
The translation of domain-level applications and domain-level constants $\tapp$ and $\tlam$ is straightforward.

The interpretation of a contextual type ${(\Phi \vdash A)}$ makes explicit the
fact that they correspond to functions $\El e \to\El \sem{A}$ where
$e = \sem{\Gamma \vdash \Phi : \ctx}$ (see Fig.~\ref{fig:ictx}).
Consequently, the corresponding
contextual object $ (\hatctx{\Phi}\vdash M)$ is interpreted as a
function which is already the case by just taking the interpretation of domain-level
terms. The case of
${(\CTV \Psi A)}$ requires the contextual object to be interpreted to the restricted function
space denoted by $\to_v$, which is the case by looking at the  variable case of
the domain-level interpretation.

 \begin{figure}[tb]
  \centering
  \begin{small}
    \[
      \begin{array}{l@{~}c@{~}ll}
\multicolumn{4}{l}{\mbox{Interpretation of contextual objects ($C$)}}\\[0.25em]
\sem{\Gamma \vdash (\hatctx{\Phi}\vdash M) : (\Phi \vdash A)} & = & \sem{\Gamma ; \Phi \vdash M : A}
\\[.25em]
\sem{\Gamma \vdash (\hatctx{\Phi}\vdash M) : (\CTV \Phi A)} & = & \sem{\Gamma ; \Phi \vdash M : A}
\\[0.5em]
\multicolumn{4}{l}{\mbox{Interpretation of contextual types ($T$)}}\\[0.25em]
\sem{\Gamma \vdash (\CT \Phi A)} &= &(\cvar{:}\El e) \to \El \sem{A} &
\mbox{where}~\sem{\Gamma \vdash \Phi : \ctx} = e
\\[0.25em]
\sem{\Gamma \vdash (\CTV \Phi A)} &= &(\cvar{:}\El e) \to_v \El \sem{A} &
\mbox{where}~\sem{\Gamma \vdash \Phi : \ctx} = e
      \end{array}
\]
  \end{small}
  \vspace{-1em}
  \caption{Interpretation of Contextual Objects and Types}
  \label{fig:ictx}
\end{figure}

Last, we give the interpretation of computation-level types,
contexts and terms (see \Cref{fig:icomp}).
It is mostly straightforward.  We simply map $\cbox{T}$ in context $\Gamma$ to $
\Box\sem{\Gamma \vdash T}$ and $\cbox{C}$ is simply interpreted as a boxed
term. Since we intend to keep all variables on the computation level crisp, we
interpret function types to the crisp function space. As a
consequence, the argument in a function application must be closed. This is true
because all computation-level terms should only use crisp variables, and it is
justified by the soundness theorem. When translating
a computation-level function, we use $\mlam$ for abstraction, so that we create
a crisp function. We then recursively interpret the function body.

\begin{figure}[tb]
  \centering
\begin{small}
\[
  \begin{array}{l@{~}c@{~}ll}
    \multicolumn{4}{l}{\mbox{Interpretation of computation-level types ($\ann\tau$)}}\\[0.25em]
    \sem{\Gamma \vdash \cbox{T}} &= & \Box\sem{\Gamma \vdash T}  \\[0.25em]
    \sem{\Gamma \vdash (x{:}\ann\tau_1) \arrow \tau_2} &= &(x{::}\sem{\Gamma \vdash
                                                            \ann\tau_1}) \to^\flat
                                                            \sem{\Gamma, x: \ann\tau_1 \vdash \tau_2}\\[0.25em]
    \sem{\Gamma \vdash \ctx} &= & \Obj
    \\[0.5em]
    \multicolumn{3}{l}{\mbox{Computation-level typing contexts ($\Gamma$)}}\\[0.25em]
    \sem{\cdot} & =  & \cdot  & \\[0.25em]
    \sem{\Gamma\comma x\of\ann\tau} & = & \sem{\Gamma}\comma x :: \sem{\Gamma \vdash \ann\tau} &
    \\[.5em]
    \multicolumn{4}{l}{\mbox{Interpretation of computations ($\Gamma\vdash t:\tau$; without recursor)}}\\[0.25em]
    \sem{\Gamma \vdash \cbox{C} : \cbox{T}} & = & \kw{box}~ e &\mbox{where}~\sem{\Gamma \vdash C : T} = e
    \\[0.25em]
    \sem{\Gamma \vdash t_1\ t_2 : \tau} & = & e_1\ e_2 & \mbox{where}~\sem{\Gamma \vdash t_1 : (x{:}\ann\tau_2) \arrow \tau} = e_1 ~
    \\
                   & & & \mbox{and}~\sem{\Gamma \vdash t_2 : \ann{\tau_2}} = e_2
    \\[0.25em]
    \sem{\Gamma \vdash  \tmfn x t : (x{:}\ann\tau_1) \arrow \tau_2}
                   & = &  \mlam x{::} \sem{\Gamma \vdash \ann\tau_1}.\, e &
\mbox{where}~\sem{\Gamma, x{:}\ann\tau_1 \vdash t : \tau_2} = e
    \\[0.25em]
    \sem{\Gamma \vdash x : \tau} & = & x  &
  \end{array}
\]
\end{small}
   \vspace{-1em}
  \caption{Interpretation of Computation-level Types and Terms -- without recursor}
  \label{fig:icomp}
\end{figure}

The interpretation of the recursor is straightforward now (see \Cref{fig:irec}). In
Lemma~\ref{lem:rec}, we expressed a primitive recursion scheme in our
internal type theory and defined a term \kw{rec} together with its
type. We now interpret every branch of our recursor in the
computation-level as a function of the required type in our internal
type theory. While this is somewhat tedious, it is straightforward. When interpreting
the branches, we again use the crisp function space, which allows us to push the
parameters to the crisp context and simulate the behavior of \cocon. The branch body
is then interpreted recursively.

\begin{figure}[tb]

\begin{small}
  \[
    \begin{array}{l@{~}c@{~}ll}
  \multicolumn{4}{l}{\mbox{Interpretation of recursor for $\IH = (\psi : \tmctx) \arrow (y:\cbox{\psi \vdash \tm}) \arrow \tau$}:}\\[.25em]
  \multicolumn{4}{l}{\sem{\Gamma \vdash  \tmrec {\IH} {b_v} {b_{\mathsf{app}}} {b_{\tlam}} \rappto \Psi~t : [{\Psi}/\psi,~t/y]\tau} = \kw{rec}~e_v~e_{\kw{app}}~e_{\kw{lam}}~e_{c}~e}\\
  \multicolumn{4}{l}{\quad\mbox{where}~\sem{\Gamma \vdash b_v : \IH} = e_v, \sem{\Gamma \vdash b_{\mathsf{app}} : \IH} = e_{\kw{app}}, \sem{\Gamma \vdash b_{\mathsf{lam}} : \IH} = e_{\kw{lam}},}\\
  \multicolumn{4}{l}{\phantom{\quad\mbox{where}}~\sem{\Gamma \vdash \Psi : \ctx} = e_{c}~\mbox{and}~\sem{\Gamma \vdash t : \cbox{\Psi \vdash \tm}} = e}\\[.5em]
  \multicolumn{4}{l}{\mbox{Interpretation of Variable Branch}}\\[.25em]
  \sem{\Gamma \vdash ({\psi,p \mto t_v}) : \IH } & = & \multicolumn{2}{l}{\mlam
                                                       \psi{::}\Obj.\,\mlam p{::}
                                                       \Box(\El \psi \to_v \El \tm
                                                       ). e}\\
  \multicolumn{4}{l}{\quad\mbox{where}~\sem{\Gamma, \psi:\tmctx, p:\cbox{\CTV{\unboxc{\psi}} {\tm}}  \vdash  t_v : [p/y]\tau } = e}
  \\[0.5em]
  \multicolumn{4}{l}{\mbox{Interpretation of Application Branch}}\\[.25em]
      \sem{\Gamma \vdash (\psi, m, n, f_n, f_m \mto t_{\mathsf{app}}) : \IH}
                                                 & = & \mlam \psi{::} \Obj.\,
                                                       \mlam m,n{::}\Box(\El
                                                       \psi \to \El \tm). \\
  & & \mlam f_m {::}\sem{\Gamma \vdash [\psi,m/\psi,y]\tau}.\, \mlam f_n{::}
      \sem{\Gamma \vdash [\psi,n/\psi,y]\tau}.e \\
  \multicolumn{4}{l}{\quad\mbox{where}~\sem{\Gamma, \psi{:}\ctx, m{:}\cbox{\psi \vdash
      \tm}, n{:}\cbox{\psi \vdash \tm}, f_m : [m/y]\tau, f_n: [n/y]\tau \vdash t_{\mathsf{app}} : [\cbox{\psi \vdash \mathsf{app}~\unbox{m}{}~\unbox{n}{}}/y]\tau } = e}
  \\[0.5em]
  \multicolumn{4}{l}{\mbox{Interpretation of Lambda-Abstraction Branch}}\\[.25em]
  \sem{\Gamma \vdash (\psi, m, f_m \mto t_{\tlam}) : \IH} & = & \mlam \psi :: \Obj.\mlam m :: \Box(\El (\kw{times}~\psi~\tm) \to \El \tm). \mlam f_m{::}\tau_m. e\\
  \multicolumn{4}{l}{\quad\mbox{where}~\sem{\Gamma \vdash [(\psi,x{:}\tm),~m/\psi, y]\tau } = \tau_m,}\\
  \multicolumn{4}{l}{\phantom{\quad\mbox{where}}~\sem{\Gamma, \psi{:}\ctx, m{:}\cbox{\psi,x{:}\tm \vdash \tm},f_m:[(\psi,x:tm),m/\psi,y]\tau \vdash t_{\mathsf{app}} : [\cbox{\psi \vdash \mathsf{lam}~\lambda x.\unbox{m}{}}/y]\tau } = e}\\
  \end{array}
  \]

\end{small}
  \vspace{-1em}
  \caption{Interpretation of Recursor}
  \label{fig:irec}
\end{figure}

We can now show that all well-typed domain-level and computation-level
objects are translated into well-typed constructions in our internal
type theory. As a consequence, we can show that equality in
\cocon implies the corresponding equivalence in our internal
type theoretic interpretation.

\begin{lemma}
  \label{lem:simplint}
  The interpretation maintains the following typing invariants:
  \begin{itemize}
  \item
    If $\Gamma \vdash \Psi \colon \ctx$ then
   $\sem{\Gamma \vdash \Psi \colon \ctx} \colon \Obj$.
  \item If $\Gamma\semi \Psi \vdash M \colon A$ then
    $\sem{\Gamma} \pipe \cdot \vdash \sem{\Gamma ; \Psi \vdash M \colon A} \colon
    (u : \El \sem{\Gamma \vdash \Psi : \ctx}) \to \El\sem{A}$.
  \item
    If $\Gamma\semi \Psi \vdash \sigma \colon \Psi$ then
    $\sem{\Gamma}\pipe \cdot \vdash \sem{\Gamma ; \Psi \vdash \sigma \colon \Psi}
    \colon (\cvar\of \El \sem{\Gamma \vdash \Psi : \ctx}) \to \El\sem{\Psi}$.
  \item
    If $\Gamma \vdash C \colon T$ then
    $\sem{\Gamma} \pipe \cdot \vdash \sem{\Gamma \vdash C : T} \colon \sem{\Gamma
      \vdash T}$.
  \item
    If $\Gamma \vdash t \colon \tau$ then
    $\sem{\Gamma} \pipe \cdot \vdash \sem{\Gamma \vdash t \colon \tau} \colon
    \sem{\Gamma \vdash\tau}$.
  \end{itemize}
\end{lemma}
The proof goes by induction on derivations. Next we show that equivalence in \cocon is
preserved by the interpretation.

\begin{proposition}[Soundness]
  The following are true.
  \begin{itemize}
  \item
    If $\Gamma\semi \Psi \vdash M \equiv N \colon A$ then\\
    $\sem{\Gamma}\pipe \cdot \vdash \sem{\Gamma \semi \Psi \vdash M : A} = \sem{\Gamma
      \semi \Psi \vdash N : A} \colon (\cvar \of \El\sem{\Psi}) \to \El\sem{A}$.
  \item
    If
    $\Gamma; \Psi \vdash \sigma \equiv \sigma' \colon \Phi$ then \\
    $\sem{\Gamma} \pipe \cdot
     \vdash
     \sem{\Gamma; \Psi \vdash \sigma : \Phi} = \sem{\Gamma; \Psi \vdash \sigma' : \Phi}
     \colon (\cvar \of \El\sem{\Psi}) \to \El \sem{\Phi}$.
  \item
    If
    $\Gamma \vdash t_1 \equiv t_2 \colon \ann\tau$ then 
    $\sem{\Gamma} \pipe \cdot \vdash \sem{\Gamma \vdash t_1 : \ann\tau} = \sem{\Gamma
      \vdash t_2 : \ann\tau} \colon \sem{\Gamma \vdash \ann\tau}$.
  \end{itemize}
\end{proposition}

The proof in the $\beta$ and $\eta$ equivalence cases of contextual types is
interesting. In the case of $\beta$ equivalence, we have
\begin{mathpar}
  \inferrule*
  {\mtyping[\Gamma][\Phi] M A \\ \mtyping \sigma \Phi}
  {\mtyping{\unbox{\cbox M}\sigma \equiv [\sigma]M}{A}}
\end{mathpar}
The right hand side is easy; substitution in terms is just composition of morphisms:
\begin{align*}
  \sem{\mtyping{[\sigma]M}{A}}
  &= e_1 \circ e_2
    \tag*{where $\sem{\mtyping[\Gamma][\Phi] M A} = e_1$ and
    $\sem{\mtyping \sigma \Phi} = e_2$}
\end{align*}

On the left hand side, we have
\begin{align*}
  \sem{\mtyping{\unbox{\cbox M}\sigma}{A}}
  &= \tletbox{x}{\kw{box}~e_1}{\lamb \cvar : \El e. x(e_2\; u)}
    \tag*{where $\sem{\isctx{\Psi}} = e$,} \\
  & \tag*{$\sem{\mtyping[\Gamma][\Phi] M A} = e_1$ and
    $\sem{\mtyping \sigma \Phi} = e_2$} \\
  &= \lamb u : \El e. e_1(e_2\ u) \tag*{due to the $\beta$ rule of $\Box$}
\end{align*}
Thus the $\beta$ equivalence of contextual types is sound semantically.

The following rule expresses the $\eta$ equivalence of contextual types:
\begin{mathpar}
  \inferrule*
  {\typing t {\cbox{\judge[\Psi]A}}}
  {\typing{t \equiv \cbox{\unbox t {\wk{\hat \Psi}}}}{\cbox{\judge[\Psi]A}}}
\end{mathpar}

We reason as follows:
\begin{align*}
  \sem{\typing{\cbox{\unbox t {\wk{\hat \Psi}}}}{\cbox{\judge[\Psi]A}}}
  &= \kw{box}~\tletbox{x}{e'}{\lamb u : \El e. x\ u}
    \tag*{where $\sem{\isctx{\Psi}} = e$ and $\sem{\typing{t}{\cbox{\judge[\Psi]A}}} =e'$} \\
  &= \kw{box}~\tletbox{x}{\sem{\typing{t}{\cbox{\judge[\Psi]A}}}}{x}
  \tag*{$\eta$ equivalence of $\lamb$}\\
  &= \sem{\typing{t}{\cbox{\judge[\Psi]A}}}
\end{align*}

The last equation requires a second thought. Indeed, this equation does not normally
hold. In fact, $\kw{box}~\tletbox{x}{t}{x} = t$ is equivalent to requiring $\Box$ to
be idempotent. In our model, this turns out to be true, as we can see from
\begin{align*}
  \Box \Box F(\Psi) = \Box F(\top) = F(\top) = \Box F(\Psi)
\end{align*}
That is, $\Box$ is \emph{definitionally} idempotent, which allows us to conclude the
equation. The fact that we rely on the idempotency of the $\Box$ modality implies that
the model can only support a two-level modal system.


\section{Dependently Typed Case}\label{sec:depty}

In the previous sections, we outlined the interpretation of a simply typed variant of
\cocon to a presheaf model. In this section, we demonstrate how the idea can be
extended to the dependently typed case. In the dependently typed case, both domain
level and computation level have dependent function spaces. This can be used to model
intrinsically simply typed languages in the domain level.
\begin{lstlisting}[escapeinside=;;]
ty : type.                         trm : ty -> type.
o : ty.                            lam : ;$\Pi$; a : ty, ;$\Pi$; b : ty, (trm a -> trm b) -> trm (arr a b).
arr : ty -> ty -> ty.              app : ;$\Pi$; a : ty, ;$\Pi$; b : ty, trm (arr a b) -> trm a -> trm b.
\end{lstlisting}
In \texttt{trm}, we use the type parameter to keep track of the object-level type of an
object-level term. We
still use HOAS to encode the case of lambda abstraction. As in the simply typed case,
we have two distinct function spaces in dependently typed \cocon as well: the weak
space that is used in the domain level and is used for HOAS, and the strong space that
is used in the computation level and supports induction.

\subsection{A Simplified \cocon}\label{sec:dep-cocon}

To model a dependently typed variant of \cocon capable of encoding this object
language, we present the modification to the syntax of the simply typed \cocon in
\Cref{fig:syntax}. In the syntax for domain-level types, we add the types for the
domain language as well as turn the simple function space into a dependent one. For
the terms, we add the corresponding constructors of domain level types, $\trm$ and $\ty$.
Compared to \citet{Pientka:LICS19}, this version of \cocon is simplified by removing
the full hierarchy of universes; as a consequence, types and terms are separated. 

\begin{figure}
  \begin{center}
    \begin{small}
      \[
        \begin{array}{p{4.8cm}@{~}l@{~}r@{~}l}
          Domain-level types           & A, B        & \bnfas & \ty \bnfalt \trm[M] \bnfalt \Pi x : A. B \\
          Domain-level terms           & M, N       & \bnfas & \lambda x.M \mid M\,N \mid x \mid \unbox t \sigma \mid c \\
          Domain-level Constants & c & \bnfas & \ttobj \pipe \ttarr \pipe \ttlam \pipe \ttapp
        \end{array}
      \]
    \end{small}
  \end{center}
  \caption{Syntax of dependent \cocon}\label{fig:syntax}
\end{figure}

\begin{figure}
  \centering\small
  \[
    \begin{array}{c}
      \multicoltext{\fbox{$\typing{\Psi}{\ctx}$}~~The domain-level context $\Psi$ is well-formed}\\
      \infer{\typing{\cdot}{\ctx}}{ }
      \qquad
      \infer
      {\typing{\psi}{\ctx}}
      {\Gamma(\psi) = \ctx}
      \qquad
      \infer
      {\typing{\Phi, x:A}{\ctx}}
      {\typing{\Phi}{\ctx} & \istype{A}}

      \\[1em]
      \multicoltext{From now on, whenever $\Psi$ presents, $\typing{\Psi}{\ctx}$ is assumed.}
      \\
      \multicoltext{\fbox{$\istype{A}$} ~~ The domain-level type $A$ is well-formed}\\
      
      \inferrule*
      { }
      {\istype{\ty}}
      \qquad
      \inferrule*
      {\mtyping{M}{\ty}}
      {\istype{\trm[M]}}
      \qquad
      \inferrule*
      {\istype{A} \\ \istype[\Gamma][\Psi, x: A]{B}}
      {\istype{\Pi x : A . B}}

      \\[2em]
      \multicoltext{\fbox{$\mtyping M A$} ~~ Term $M$ has type $A$ in domain-level
      context $\Psi$ and context $\Gamma$} \\

      \inferrule*
      {x : A \in \Psi}
      {\mtyping x A}
      \qquad
      \inferrule*
      {\mtyping[\Gamma][\Psi, x : A] M B}
      {\mtyping{\lambda x.M}{\Pi x : A . B}}
      \qquad
      \inferrule*
      {\mtyping M {\Pi x : A. B} \\ \mtyping N A}
      {\mtyping{M\;N}{[N/x]B}}

      \\[0.5em]

      \inferrule*
      {\typing t {\quot{\judge[\Phi] A}}\text{ or }\typing t {\quot{\vjudge \Phi A}} \\ \mtyping \sigma \Phi}
      {\mtyping{\unquote t \sigma}{[\sigma]A}}

      \qquad

      \inferrule*
      {\mtyping M A \\ \istype{A \equiv B}}
      {\mtyping M B}
      
      \qquad
      
      \inferrule*
      { }
      {\mtyping{\ttobj}{\ty}}

      \\[0.5em]
      
      \inferrule*
      {\mtyping{a}{\ty} \\ \mtyping{b}{\ty}}
      {\mtyping{\ttarr\;a\;b}{\ty}}

      \qquad
      
      \inferrule*
      {\mtyping{a}{\ty} \\ \mtyping{b}{\ty} \\ \mtyping[\Gamma][\Psi] f {\trm[a] \to \trm[b]}}
      {\mtyping{\ttlam\;a\;b\;f}{\trm[(\ttarr\;a\;b)]}}

      \\[0.5em]
      
      \inferrule*
      {\mtyping{a}{\ty} \\ \mtyping{b}{\ty} \\ \mtyping m {\trm[(\ttarr\;a\;b)]} \\ \mtyping n {\trm[a]}}
      {\mtyping{\ttapp\;a\;b\;m\;n}{\trm[b]}}

      \\[2em]
      \multicoltext{\fbox{$\mtyping \sigma \Phi$} ~~ Substitution $\sigma$ provides a mapping from the (domain) context $\Phi$ to $\Psi$} \\

      \inferrule*
      {\isctx{\Psi}}
      {\mtyping \cdot \cdot}
      \qquad
      \inferrule*
      {\mtyping \sigma \Phi \\ \mtyping M {[\sigma]A}}
      {\mtyping{\sigma, M}{\Phi, x : A}}
      \qquad
      \inferrule*
      {\isctx{\Psi, \overrightarrow{x : A}}}
      {\mtyping[\Gamma][\Psi, \overrightarrow{x : A}]{\wk{\hat\Psi}}\Psi}
    \end{array}
  \]
  \caption{Domain-level judgments}\label{fig:lfjudgments}
\end{figure}

The typing rules are shown in \Cref{fig:lfjudgments} and are changed more
significantly. Since domain-level terms can appear in types now, we need to add
well-formedness condition for domain-level contexts, $\isctx \Psi$, and types,
$\istype A$. In particular, $\trm[M]$ is well-formed only if $M$ has type $\ty$. In the
typing rules, the application rule and the unbox rule shows how dependent types are
involved. In particular, in the unbox case, the substitution $\sigma$ is applied to
$A$ as the resulting type. Without loss of generality, we require object-level
constructors to be fully applied, e.g. $\ttarr\ a$ is not a valid domain-level term. It
helps to simplify the semantic interpretation and allows us to focus on the essential
idea of the development.

\begin{figure}
  \centering\small
  \[
    \begin{array}{c}
      \multicoltext{\fbox{$\vdash \Gamma$}~~ $\Gamma$ is a well-formed context}
      \\
      \inferrule*
      { }
      {\vdash \cdot}
      \qquad
      \inferrule*
      {\vdash \Gamma \\ \Istype{\ann\tau}}
      {\vdash \Gamma, x : \ann\tau}
      \\[1em]
      \multicoltext{\fbox{$\Istype{\ann\tau}$}~~ $\ann\tau$ is a well-formed type in context $\Gamma$}
      \\
      \inferrule*
      { }
      {\Istype \ctx}
      \qquad
      \inferrule*
      {\Istype{\ann\tau_1} \\ \Istype[\Gamma, y: \ann\tau_1]{\tau_2}}
      {\Istype{(y : \ann \tau_1) \arrow \tau_2}}
      \qquad
      \inferrule*
      {\istype A}
      {\Istype{\cbox{\Psi \vdash A}}}
      \qquad
      \inferrule*
      {\istype A}
      {\Istype{\cbox{\Psi \vdash_v A}}}
      \\[1em]
      \multicoltext{\fbox{$\typing C T$}~~ Contextual object $C$ has contextual type
      $T$}
      \\
      \inferrule*
      { }
      {\typing{(\judge[\hat{\Psi}]M)}{(\judge[\Psi]A)}}
      \qquad
      \inferrule*
      {x : A \in \Psi}
      {\typing{(\judge[\hat{\Psi}]x)}{(\vjudge \Psi A)}}
      \qquad
      \inferrule*
      {x : \quot{\vjudge \Phi A} \in \Gamma \\
      \mtyping{\wk{\hat{\Psi}}}\Phi}
      {\typing{(\judge[\hat{\Psi}]\unquote{x}{\wk{\hat{\Psi}}})}{(\vjudge \Psi A)}}

      \\[1em]
      \multicoltext{\fbox{$\typing t {\ann\tau}$} ~~ Term $t$ has computation type $\ann\tau$}
      \\      
      \inferrule*
      {y : \ann\tau \in \Gamma}
      {\typing y \ann{\tau}}
      \qquad
      \inferrule*
      {\typing C T}
      {\typing{\quot C}{\quot T}}
      \qquad
      \inferrule*
      {\typing[\Gamma, y:\ann{\tau_1}] t {\tau_2}}
      {\typing{\deffun y t}{(y : \ann{\tau_1}) \arrow \tau_2}}
      \\[1em]
      \inferrule*
      {\typing t {(y : \ann{\tau_1}) \Rightarrow \tau_2} \\ \typing s {\ann{\tau_1}}}
      {\typing{t\;s}{[s/y]\tau_2}}
      \qquad
      \inferrule*
      {\typing t{\ann\tau_1} \\ \Istype{\ann\tau_1 \equiv \ann\tau_2}}
      {\typing t {\ann\tau_2}}

      \\[1.5em]
      \multicoltext{Recursor over $\ty$: $\mathcal{I} = (\psi : \ctx) \Rightarrow (y :
      \quot{\judge[\psi]{\ty}}) \Rightarrow \tau$} \\
      \inferrule*
      {\typing{t}{\quot{\judge[\Psi]\ty}} \\
      \typing{b_\ttobj}{\mathcal{I}} \\ \typing{b_{\ttarr}}{\mathcal{I}}}
      {\typing{\induct{\mathcal{I}}{(b_\ttobj \pipe b_{\ttarr})}{\Psi}{t}}{[\Psi,t/\psi, y]\tau}}

      \\[1em]
      
      \text{Branch for the $\ttobj$ case} \quad
      \inferrule*
      {\typing[\Gamma, \psi : \ctx]{t_\ttobj}{[\quot{\judge[\psi]\ttobj} / y]\tau}}
      {\typing{(\psi \mapsto t_\ttobj)}{\mathcal{I}}}
      \\[1em]

      \text{Branch for the $\ttarr$ case}
      
      \inferrule*
      {\typing[\Gamma, \psi : \ctx, m, n : \quot{\judge[\psi]\ty}, f_m : [m/y]\tau, f_n :
      [n/y]\tau]{t_{arr}}{[\quot{\judge[\psi]{\ttarr\unquot{m}\unquot{n}}}/y]\tau}}
      {\typing{(\psi, m, n, f_m, f_n \mapsto t_{\ttarr})}{\mathcal{I}}}

      \\[1.5em]
      \multicoltext{
      Recursor over $\trm$: $\mathcal{I} = (\psi : \ctx) \Rightarrow (z :
      \quot{\judge[]{\ty}}) \Rightarrow (y :
      \quot{\judge[\psi]{\trm[\unquote{z}{\cdot}]}}) \Rightarrow \tau$} \\

      \inferrule*
      {\typing{t}{\quot{\judge[]\ty}} \\ \typing{t'}{\quot{\judge[\Psi]\trm[\unquote{t}{\cdot}]}} \\
      \typing{b_v}{\mathcal{I}} \\ \typing{b_{\ttlam}}{\mathcal{I}} \\
      \typing{b_{\ttapp}}{\mathcal{I}}}
      {\typing{\induct{\mathcal{I}}{(b_v \pipe b_{\ttlam} \pipe
      b_{\ttapp})}{\Psi}{t}\;t'}{[\Psi,t,t'/\psi, z, y]\tau}}

      \\[1em]
      
      \text{Branch for the variable case} \quad
      \inferrule*
      {\typing[\Gamma, \psi : \ctx, a : \quot{\judge[]\ty}, t :
      \quot{\vjudge{\psi}{\trm[\unquote{a}{\cdot}]}}]{t_v}{[a, t/ z, y]\tau}}
      {\typing{(\psi, a, t \mapsto t_v)}{\mathcal{I}}}

      \\[1em]
      \text{Branch for the $\ttlam$ case}
      
      \inferrule*
      {\Gamma, \psi: \ctx, a, b : \quot{\judge[]\ty},  
      m : \quot{\judge[\psi, x : \trm[\unquote{a}{\cdot}]]{\trm[\unquote{b}{\cdot}]}}, \\
      f_m : [(\psi, \trm[\unquote{a}{\cdot}]), b, m/\psi, x, y]\tau \vdash {t_{\ttlam}}
      : {[\quot{\judge[]{\ttarr\unquot{a}\unquot{b}}},
      \quot{\judge[\psi]{\ttlam\unquote{a}{\cdot}\unquote{b}{\cdot}(\lambda
      x. \unquot m)}}/ z, y]\tau}}
      {\typing{(\psi, a, b, m, f_m \mapsto t_{\ttlam})}{\mathcal{I}}}
      
      \\[1em]
      \text{Branch for the $\ttapp$ case}
      
      \inferrule*
      {\Gamma, \psi : \ctx, a, b : \quot{\judge[]\ty},
      m : \quot{\judge[\psi]{\trm{(\ttarr\unquote{a}{\cdot}\unquote{b}{\cdot})}}},
      n : \quot{\judge[\psi]{\trm{\unquote{a}{\cdot}}}}, \\
      f_m : [\quot{\judge[]{\ttarr\unquot{a}\unquot{b}}}, m/ z, y]\tau,
      f_n : [a,n/z, y]\tau \vdash t_{\ttapp} : {[b,\quot{\judge[\psi]{\ttapp\unquote{a}{\cdot}\unquote{b}{\cdot}\unquot m
      \unquot n}}/z, y]\tau}
      }
      {\typing{(\psi, a, b, m, n, f_m, f_n \mapsto t_{\ttapp})}{\mathcal{I}}}
      
    \end{array}
  \]
  
  \caption{Judgments in the computation level}\label{fig:compjudgments}
\end{figure}

The computation-level judgments are shown in \Cref{fig:compjudgments}. Similar to the
domain level, we also need well-formedness judgments for computation-level contexts
and types. Unlike \cocon defined in \citet{Pientka:LICS19}, we need a separate
judgment for well-formed types because we do not have universes here. The
computation-level language is also dependently typed, as shown in the application
rule. We also formulate the induction principle for $\trm$ and $\ty$. The case for
$\ty$ is straightforward as it is a normal algebraic data type. For $\trm$, in
addition to the cases of $\ttlam$ and $\ttapp$, we need a case for variables as in the
simply typed settings.

\subsection{Categories with Families}

In the previous section, we showed that we can regard the domain-level language 
as a cartesian closed category~$\CC$ and use the Yoneda embedding to embed the domain
level into the presheaf category $\hat\CC$, which is regarded as the computation-level
language. Interestingly, this model can be extended to the case of dependently typed
domain languages. In this case, $\CC$ needs to have enough
structure to model dependent types, and the model we consider here is categories with
families (CwFs)~\citep{DBLP:conf/types/Dybjer95,Hofmann:NI97}.
\begin{definition}\label{def:cwa}
  A \emph{category with families} $\mathcal{C}$ consists of the following data:
  \begin{enumerate}
  \item a terminal object $\top$,
  \item a functor $\textsf{Ty} : \mathcal{C}^{op} \to \Set$, whose action on morphisms we denote as $-\{\sigma\} : \Ty \Phi \to \Ty \Psi$ for
    $\sigma : \Psi \to \Phi$,
  \item for $\Phi \in \mathcal{C}$ and $A \in \Ty \Phi$, a set $\Tm \Phi A$, such that:
    \begin{itemize}
      \item for $\sigma : \Psi \to \Phi$ and $t \in \Tm \Phi A$, $t \{ \sigma \} \in \Tm{\Psi}{A\{\sigma\}}$, and
      \item the equations $t\{id_\Phi\} = t$ and  $t\{\sigma\}\{\delta\} = t \{ \sigma \circ \delta \}$ are valid.
    \end{itemize}
  \item a context comprehension $-.-$ so that given $\Phi \in \mathcal{C}$ and $A \in
    \Ty \Phi$, $\Phi.A \in \mathcal{C}$,
  \item for $\Phi \in \mathcal{C}$ and $A \in \Ty \Phi$, a projection morphism of context comprehension $p(A) : \Phi.A \to \Phi$, 
  \item for $\Phi \in \mathcal{C}$ and $A \in \Ty \Phi$, a variable projection $v_A
    \in \Tm{\Phi.A}{A\{p(A)\}}$, and
  \item for $\sigma : \Psi \to \Phi$, $A \in \Ty \Phi$, and $t \in
    \Tm{\Psi}{A\{\sigma\}}$, a unique extension morphism $\langle \sigma, t \rangle :
    \Psi \to \Phi.A$. 
  \end{enumerate}

  The following equations hold:
  \begin{enumerate}
  \item $p(A) \circ \langle \sigma, t \rangle = \sigma$,
  \item $v_A\{\langle \sigma, t \rangle\} = t$, and
  \item $\sigma = \langle p(A) \circ \sigma , v_T\{\sigma\} \rangle$ where $\sigma :
    \Phi \to \Psi.A$.
  \end{enumerate}
\end{definition}

We regard the objects in $\mathcal{C}$ as contexts and the morphisms as substitutions
of one context for another. Based on this understanding, $\Ty \Phi$ denotes the set of
semantic types in context $\Phi$. Given a semantic type $A \in \Ty \Phi$,
$A\{\sigma\}$ is thus applying the substitution $\sigma$ to $A$ given
$\sigma : \Psi \to \Phi$. Based on their definitions, we can prove the following
properties of type and term substitutions. Given $A \in \Ty \Phi$, $\sigma : \Psi \to
\Phi$, $\delta : \Psi' \to \Psi$, $t : \Tm \Phi A$:
\begin{align*}
  A\{id_\Phi\} &= A \\
  A\{\sigma\}\{\delta\} &= A \{ \sigma \circ \delta \}
\end{align*}

$\Tm \Phi A$ denotes the set of semantic terms of type $A$ in context $\Phi$. Given a
term $t \in \Tm \Phi A$ and a substitution $\sigma : \Psi \to \Phi$, $t\{\sigma\} :
\Tm{\Psi}{A\{\sigma\}}$ is
the result of applying $\sigma$ to $t$. Note that the type of this term is
$A\{\sigma\}$, so CwFs are capable of handling dependent types. 

Sometimes, given a substitution $\sigma : \Psi \to \Phi$ and $A \in \Ty\Phi$, we
would like to obtain another substitution $q(\sigma, A) : \Psi.A\{\sigma\} \to
\Phi.A$. This substitution is needed below when we define substitution for $\Pi$
types. We can define
\begin{align*}
  q(\sigma, A) := \langle \sigma \circ p(A\{\sigma\}), v_{A\{\sigma\}} \rangle
\end{align*}
By applying the property of $p(A\{\sigma\})$, we can see that the following diagram
is a pullback:
\begin{center}
  \begin{tikzcd}
    \Psi.A\{\sigma\} \arrow[rr, "{q(\sigma,A)}"] \arrow[dd, "p(A\{\sigma\})"'] &  & \Phi.A \arrow[dd, "p(A)"] \\
    &  &                           \\
    \Psi \arrow[rr, "\sigma"']                                                 &  & \Phi
  \end{tikzcd}
\end{center}

We shall work with telescopes of types. A \emph{telescope} of types in context $\Phi$ is a sequence 
of types $A_1,A_2,\dots, A_n$ such that $A_1\in\Ty{\Phi}$, $A_2\in\Ty{\Phi.A_1}$,
\dots, $A_n\in\Ty{\Phi.A_1.\dots A_{n-1}}$.
We write $\vec A$ to range over telescopes and extend context comprehension, substitution
and projection to telescopes in the canonical way. That is, we write $\Phi.\vec A$ for 
$\Phi.A_1.\dots.A_n$ and $p(\vec A)$ for $p(A_1)\circ \dots \circ p(A_n)$.


Up until this point, we have obtained a generic categorical structure for dependent type
theory. In order to model the dependent function space, we need a semantic type former.
\begin{definition}~\citep{Hofmann:NI97}
  Semantic $\Pi$ types in a CwF $\mathcal{C}$ have the following data:
  \begin{enumerate}
  \item A semantic type $\Pi(A, B) \in \Ty\Phi$ for each $A \in \Ty\Phi$ and $B \in
    \Ty{\Phi.A}$,
  \item a semantic term $\Lambda_{A,B}(M) \in \Tm{\Phi}{\Pi(A, B)}$ for each $M \in
    \Tm{\Phi.A}B$, and
  \item a semantic term $\App_{A,B}(M, N) \in \Tm\Phi{B\{\ctxext{id_\Phi, N}\}}$ for each $M \in \Tm\Phi
    {\Pi(A, B)}$ and $N \in \Tm\Phi A$.
  \end{enumerate}
  so that the following axioms are satisfied:
  \begin{enumerate}
  \item $\Pi(A, B)\{\sigma\} = \Pi(A\{\sigma\}, B\{q(\sigma, A)\}) \in \Ty\Psi$ for
    $\sigma : \Psi \to \Phi$,
  \item
    $\Lambda_{A,B}(M)\{\sigma\} = \Lambda_{A\{\sigma\},B\{q(\sigma, A)\}}(M\{q(\sigma,
    A)\}) \in \Tm\Psi{\Pi(A, B)\{\sigma\}}$ for $M \in \Tm{\Phi.A}B$ and
    $\sigma : \Psi \to \Phi$,
  \item
    $\App_{A,B}(M, N)\{\sigma\} = \App_{A\{\sigma\},B\{q(\sigma, A)\}}(M\{\sigma\},
    N\{\sigma\}) \in \Tm\Psi{B\{\ctxext{id_\Phi, N}\}\{\sigma\}} = \Tm\Psi
    {B\{\ctxext{\sigma, N\{\sigma\}}\}}$ for $M \in \Tm\Phi{\Pi(A, B)}$,
    $N \in \Tm\Phi A$ and $\sigma : \Psi \to \Phi$,
  \item
    $\App_{A,B}(\Lambda_{A,B}(M), N) = M\{\ctxext{id_\Phi, N}\} \in \Tm\Phi{B\{\ctxext{id_\Phi,
      N}\}}$ for each $M \in \Tm{\Phi.A}B$ and $N \in \Tm\Phi A$, and
  \item $\Lambda_{A,B}(\App_{A,B}(M\{p(A)\}, v_A)) = M \in \Tm\Phi{\Pi(A, B)}$.
  \end{enumerate}
\end{definition}
We often omit the subscripts of $\Lambda$ and $\App$ in favor of conciseness when they can be
unambiguously inferred. For example, the fourth axiom can be more concisely expressed
as $\App(\Lambda(M), N) = M\{\ctxext{id_\Phi, N}\}$.

One characteristic of this framework is that the consistency of equations involves
reasoning about equality between sets. Consider the third equation above. We can
transform the left hand side as follows:
\begin{align*}
  \App(M, N)\{\sigma\} &\in \Tm{\Phi}{B\{\ctxext{id_\Psi, N}\}\{\sigma\}} \\
                      &= \Tm{\Phi}{B\{\ctxext{id_\Psi, N} \circ \sigma\}} \\
                      &= \Tm{\Phi}{B\{\ctxext{\sigma, N\{\sigma\}}\}}
                        \tag*{property of extension morphism}
\end{align*}
The right hand side has:
\begin{align*}
  \App(M\{\sigma\}, N\{\sigma\}) &\in \Tm{\Phi}{B\{q(\sigma, A)\}\{\ctxext{id_\Phi,
                                  N\{\sigma\}}\}} \\
                                &= \Tm{\Phi}{B\{q(\sigma, A) \circ \ctxext{id_\Phi,
                                  N\{\sigma\}}\}} \\
                                &= \Tm\Phi {B\{\langle \sigma \circ p(A\{\sigma\}), v_{A\{\sigma\}} \rangle \circ \ctxext{id_\Phi,
                                  N\{\sigma\}}\}} \\
                                &= \Tm\Phi{B\{\ctxext{\sigma, N\{\sigma\}}\}}
\end{align*}
As both terms belong to the same set, the equation is well-defined.

Combining the definition of a CwF with $\Pi$ types, we are able to capture the nature
of dependent types with dependent function spaces in both the domain level and the computation
level. We can also axiomatize other types like $\ty$ and $\trm$ and their
constructors. However, we need more to interpret \cocon:
\begin{enumerate}
\item CwFs do not provide a direct connection between the domain and the computation
  level;
\item the domain-level terms exhibit different substitution behavior from the computation
  terms.
\end{enumerate}
To overcome these issues, we can extend the model presented in \Cref{sect:presheaves}
by requiring the domain category to possess structure of a CwF. We will expand our
discussion in the next section.

\subsection{Presheaves over a Small Category with Families}

Having introduced an appropriate notion of model for dependent type theories,
we can now generalise the construction from \Cref{sect:simple} to the case of
dependent domain languages. 

As in the simply-typed case, we begin with a term model~$\CC$ of the
domain-level type theory. The main difference is that in the dependent case,
this category has the structure of a CwF instead of being just cartesian closed.
In order to interpret the computation types of \cocon, we work in the presheaf
category $\hat\CC$. This category has enough structure to interpret \cocon
computations and also embeds~$\CC$ fully and faithfully via the Yoneda embedding.

\paragraph{A Universe for a small CwF} For working with the internal type theory
of~$\hat \CC$, it is again convenient to capture the embedding of $\CC$ into
$\hat\CC$ in terms of a Tarski-style universe. It is given by the following term
and type constants in the internal type theory of~$\hat \CC$. The types $\Ctx$
and $\El \Phi$ generalize $\Obj$ and $\El a$ from \Cref{sect:simple} to the
dependent case: \begin{mathpar}
  
  \Istype[\cdot]{\Ctx}

  \Istype[\Phi : \Ctx]{\El(\Phi)}

  \Istype[\Phi : \Ctx]{\tTy\Phi}

  \Istype[\Phi : \Ctx\comma A : \tTy\Phi]
  {\tTm\Phi A}
  
  \typing[\cdot]{\top}{\Ctx}

  \typing[\Phi : \Ctx\comma A : \tTy\Phi]{\Phi.A}{\Ctx}

  \typing[\Phi : \Ctx]{{!}}{\El(\Phi) \to \El(\top)}

  \typing[\Phi : \Ctx\comma A : \tTy\Phi]{p}{\El(\Phi.A) \to \El(\Phi)}

  \typing[\Phi : \Ctx\comma A : \tTy\Phi]{v}{\tTm{\Phi.A}{A\{p\}}}

  \typing[\Phi : \Ctx\comma \Psi : \Ctx\comma A : \tTy \Psi \comma \sigma : \El(\Phi) \to \El(\Psi)]
  {A\{\sigma\}}{\tTy\Phi}

  \typing[\Phi,\Psi : \Ctx\comma A : \tTy\Psi \comma M : \tTm\Psi A\comma \sigma : \El(\Phi) \to \El(\Psi)]
  {M\{\sigma\}}{\tTm\Phi{A\{\sigma\}}}

    \typing[\Psi, \Phi : \Ctx\comma \sigma : \El(\Psi) \to \El(\Phi)\comma A :
  \tTy\Phi\comma M : \tTm\Psi{A\{\sigma\}}]{\ctxext{\sigma, M}}{\El(\Psi) \to
    \El(\Phi.A)}
  
\end{mathpar}

Let us outline the interpretation of these constants in~$\hat\CC$ next.
Recall from \Cref{sec:internal} that contexts are objects of $\hat\CC$, 
types in context~$\Gamma$ are presheaves $\hTy\Gamma = \elop\Gamma \to \Set$, and terms
are sections of the projections maps $p\colon \Gamma.A \to \Gamma$ in 
$\hat \CC$. We detail the required structure used in the interpretation.

The interpretation of the types $\Ctx$ and $\El(\Phi)$ is as follows:
\begin{align*}
  \Ctx &:  \hTy\top
  \\
  \Ctx(\Psi,*) &= \{ \vec A \mid \text{$\vec A$ is a telescope of types in context~$\Psi$}\}
  \\  
  \Ctx(\sigma:(\Psi',*) \to (\Psi,*)) &= \vec A \mapsto \vec A\{\sigma\}
\end{align*}
\begin{align*}
  \El &:  \hTy{\top.\Ctx}
  \\
  \El(\Psi,(*,\vec{A})) &= \{ \sigma \in \CC(\Psi, \Psi.\vec{A}) \mid p(\vec A)\circ \sigma = \id_\Psi \}
  \\  
  \El(\sigma: (\Psi', (*,\vec{A}\{\sigma\}) \to (\Psi, (*,\vec{A}))) &= f \in \CC(\Psi, \Psi.\vec{A}) \mapsto f\{\sigma\}
\end{align*}

The types of the form $\Ty \Phi$ in the internal type theory are interpreted
as follows:
\begin{align*}
  \texttt{Ty} &: \hTy{\top.\Ctx}
  \\
  \tTy{\Psi,(*,\vec{A})} &= \Ty{\Psi.\vec{A}}
  \\
  \tTy{\sigma:(\Psi', (*,\vec{A}\{\sigma\}) \to (\Psi, (*,\vec{A}))} &= B\in \Ty{\Psi.\vec{A}} \mapsto B\{q(\sigma, \vec{A})\}
\end{align*}
Here $\textsf{Ty}$ on the right hand side is given by the CwF structure of the domain
category $\CC$. That is, $\texttt{Ty}$ is defined in terms of $\textsf{Ty}$ in the domain
category, so all types of the domain level can be referred to as terms in the presheaf
category. 

The types of the form $\tTm{\Phi}{A}$ are interpreted as follows.
\begin{align*}
  \texttt{Tm} &: \hTy{\top.\Ctx.\texttt{Ty}}
  \\
  \texttt{Tm}(\Psi,(*,\vec A, B)) &= \Tm{\Psi.\vec A}{B}
  \\
  \texttt{Tm}(\sigma) &= M\mapsto M\{q(\sigma, \vec A)\}
\end{align*}
Similar to $\texttt{Ty}$, $\texttt{Tm}$ is also defined by terms
$\CTm$ in
the domain category $\CC$. 

For spelling out the interpretation of the remaining terms, we need to give
a manageable presentation of the semantic interpretation of function types
of the form $\El(\Phi) \to \El(\Phi')$. In the simply typed case, the Yoneda lemma
shows an isomorphism between $\CC(\Phi, \Phi')$ and $\hat\CC(\y(\Phi), \y(\Phi'))$. In the dependently typed case, this
isomorphism remains: the function space $\El(\Phi) \to \El(\Phi')$ is isomorphic to
$\CC(\Phi, \Phi')$, as proved by the following lemma:
\begin{lemma}\label{lem:explicitfun}
  The interpretation of
  \[
      \Istype[\Phi\of \Ctx, \Phi' \of \Ctx]{\El(\Phi) \to \El(\Phi')}
  \]
  in the internal type theory of\/~$\hat\CC$ is isomorphic to the following
  presheaf $\textsf{Hom}$:
  \begin{align*}
    \textsf{Hom} &\colon \hTy{\top.\Ctx.\Ctx\{p\}}
    \\  
    \textsf{Hom}(\Psi, (*,\vec{A}, \vec{B})) &= \{ \sigma \in \CC(\Psi.\vec{A}, \Psi.\vec{B}) \mid p(\vec{B}) \circ \sigma = p(\vec{A}) \}
  \end{align*}
\end{lemma}
\begin{proof}
  For one direction, as given in \Cref{sec:internal}, the type $\El(\Phi) \to
  \El(\Phi')$ in the syntax is interpreted to a function $f$ which takes $\sigma \colon
  \Phi \to \Psi$ and $M:\Phi \to \Phi.\vec A\{\sigma\}$ and outputs $\Phi \to
  \Phi.\vec B\{\sigma\}$ that is natural in $\Phi$. Then we obtain $f(p_{\vec A} :
  \Psi.\vec A \to \Psi, v : \Psi.\vec A \to  \Psi.\vec A. \vec A \{p_{\vec A}\}) :
  \Psi.\vec A \to \Psi.\vec A.\vec B\{p_{\vec A}\}$, such that $p_{\vec A} \circ
  f(p_{\vec A}, v) = \id_{\Psi.\vec A}$. We obtain $\Psi.\vec A \to \Psi.\vec B$ by
  composing $q(p_{\vec A}, \vec B)$ with $f(p_{\vec A}, v)$. 

  The other direction is given $\sigma : \Psi.\vec A \to \Psi.\vec B$, $\delta : \Phi
  \to \Psi$ and $\sigma' : \Phi \to \Phi.\vec A\{\delta\}$ and outputs $\Phi \to
  \Phi.\vec B\{\delta\}$.
  \begin{center}
\begin{tikzcd}
                          & \Psi.\vec A \arrow[r, "\sigma"]                                   & \Psi.\vec B                                             \\
\Phi \arrow[r, "\sigma'"] & \Phi.\vec A\{ \delta\} \arrow[u, "{q(\delta, \vec A)}"] \arrow[r] & \Phi.\vec B\{\delta\} \arrow[u, "{q(\delta, \vec B)}"']
\end{tikzcd}
  \end{center}
  Thus the morphism $\Phi \to \Phi.\vec B\{\delta\}$ is obtained by composing~$\sigma'$ 
    with the unique morphism that makes the above square a pullback. 
\end{proof}
In the following, we shall apply the isomorphism from the lemma implicitly and treat the
interpretation of $\El(\Phi) \to \El(\Phi')$ to be the same as $\textsf{Hom}$.

With this convention, we are able to give an explicit interpretation for the judgments that
we give at the beginning of this section. Recall, for example, the term for type substitution:
\begin{align*}
    \typing[\Phi : \Ctx\comma \Psi : \Ctx\comma A : \tTy \Psi \comma \sigma : \El(\Phi) \to \El(\Psi)]
  {A\{\sigma\}}{\tTy\Phi}
\end{align*}
This term $A\{\sigma\}$ is 
interpreted as the following natural transformation:
\begin{align*}
  A\{\sigma\}_\Theta & :=
    (*,\vec{A}, \vec{B}, C : \tTy{\Theta.\vec{B}}, \sigma : \Theta.\vec{A} \to \Theta.\vec{B})
    \mapsto
    (*,\vec{A}, \vec{B}, C, \sigma, C\{\sigma\})
\end{align*}
Recall that terms of $\hat\CC$ are interpreted as sections of context 
projections, so $A\{\sigma\}$ is a natural transformation from
the interpretation of the context 
$\Gamma := \Phi : \Ctx\comma \Psi : \Ctx \comma A : \tTy \Psi \comma \sigma : \El(\Phi) \to \El(\Psi)$
to the interpretation of the context 
$\Gamma \comma B : \tTy \Phi$. $\vec A$ and $\vec B$ interpret $\Phi$ and $\Psi$
respectively because $\Ctx$ is interpreted as a set of telescopes of $\Theta$. 



We omit the interpretation of the remaining terms and the verification of the equations and next turn to the case
where $\CC$ is a CwF with dependent product.

\paragraph{Domain-level $\Pi$ types}


Domain-level $\Pi$ types can be formulated by the following constants:
\begin{mathpar}
  \typing[\Phi : \Ctx, A : \tTy \Phi, B : \tTy{\Phi.A}]{\Pi(A, B)}{\tTy\Phi}

  \typing[\Phi : \Ctx, M : \tTm{\Phi.A} B]{\Lambda(M)}{\tTm\Phi {\Pi(A, B)}}

  \typing[\Phi : \Ctx, M : \tTm\Phi {\Pi(A, B)}, N : \tTm\Phi A]{\App(M, N)}{\tTm\Phi {B\{\ctxext{1_\Phi, N}\}}}
\end{mathpar}
$\Pi$ types defined here reside in the presheaf category and model $\Pi$ types in the
domain category $\CC$. They can be defined in terms of $\Pi$ types in $\CC$ as the
types defined above:
\begin{align*}
  \Pi_\Psi := (*, \vec A, C : \Ty{\Psi.\vec A}, D : \Ty{\Psi.\vec A.C}) \mapsto (*,
  \vec A, C : \Ty{\Psi.\vec A}, D : \Ty{\Psi.\vec A.C}, \Pi(C, D))
\end{align*}
where $\Pi(C, D)$ is given by the $\Pi$ types in $\CC$. 


\paragraph{Object-level Language}

Finally we can define a model for the object-level language defined in the beginning of
\Cref{sec:depty}. In this object language, $\trm$ is indexed by $\ty$. We encode the
object language as follows:
\begin{mathpar}
  \typing[\cdot]{\tty}{\tTy\top}

  \typing[\cdot]{\ttrm}{\tTy{\top.\tty}}

  \typing[\Phi:\Ctx]{\tty' := \tty\{!\}}{\tTy\Phi}
\end{mathpar}
For convenience, we introduce the abbreviation $\tty'$ for $\tty\{!\}$. We model $\ty$
and $\trm$ in their minimal contexts in order to avoid formulating their coherence
conditions w.r.t. substitutions. In the front-end language, we use the applied form
$\trm[a]$ to denote a type of terms with type $a$. In the model given as here, we
must apply a substitution instead:
\begin{mathpar}
  \typing[\Phi : \Ctx, a : \tTm\Phi{\tty'}]{\ttrm\{\ctxext{!, a}\}}{\tTy\Phi}
\end{mathpar}
That is, in the front-end syntax, we write $\trm~a$, which is interpreted as
$\ttrm\{\ctxext{!, a}\}$ in the model. We introduce an abbreviation for later:
\begin{align*}
  \ttrm[a] := \ttrm\{\ctxext{!, a}\}
\end{align*}

Finally, the constructors of the object language are formulated as follows:
\begin{mathpar}
  \typing[\Phi : \Ctx]{\tobj}{\tTm\Phi{\tty'}}

  \typing[\Phi : \Ctx] {\tarr} {\tTm\Phi{\tty'} \to \tTm{\Phi}{\tty'} \to
    \tTm{\Phi}{\tty'}}

  \typing[\Phi : \Ctx]
  {\tlam}
  {(a, b : \tTm{\Phi}{\tty'}) \to
     \tTm{\Phi.\ttrm[a]}{\ttrm[b\{p\}]} \to
    \tTm{\Phi}{\ttrm[\tarr(a, b)]}}

  \typing[\Phi : \Ctx]
  {\tapp}
  {(a, b : \tTm{\Phi}{\tty'}) \to
    \tTm{\Phi}{\ttrm[\tarr(a, b)]} \to
    \tTm{\Phi}{\ttrm[a]} \to \tTm\Phi{\ttrm[b]}}
\end{mathpar}

\paragraph{Category-theoretic perspective}

A number of properties of the universe~$\El$ can be obtained by
category-theoretic considerations. We have explained~$\El$ as a syntactic
representation of the Yoneda embedding. It has been
shown~\cite{capriotti:thesis} that the Yoneda embedding is a morphism of CwFs,
which means that it is a functor preserving the CwF structure. Using the
notation of \cite[Definition~2.1.4]{capriotti:thesis} this means that there are
isomorphisms such as $\y(\Phi.A) \iso \y(\Phi).\y^{\mathrm{Ty}}(A)$ for the
preservation of context comprehension. The terms and types for the
universe~$\El$ in the internal type theory that we have defined at the beginning
of this section are a syntactic presentation of this structure.

With this category-theoretic view, it is possible to use existing results on
morphisms of CwFs to obtain information about the universe~$\El$. For example,
\cite[Proposition~4.8]{DBLP:journals/mscs/ClairambaultD14} can be used to show
that the Yoneda embedding preserves $\Pi$-types up to isomorphism, because the
Yoneda embedding preserves local cartesian closed structure~\cite[Lemma~4.5]{DBLP:conf/ctcs/Pitts87}.
This gives us isomorphisms such as between $\El(\Phi) \to \El(\Phi')$ and
$\El(\Phi \to \Phi')$. For instance the direct proof of
Lemma~\ref{lem:explicitfun} above can be understood as an instance of this
isomorphism.

\subsection{Interpreting the Domain Level}

Given the semantic model, we can detail the interpretation of the dependently typed
\cocon defined in \Cref{sec:dep-cocon} into this model. The interpretation is a
natural generalization of the simply typed version. First we will consider the
interpretation of the domain-level types and terms, as shown in \Cref{fig:intplf}. One
complication we encountered here is that various judgments are interdependent. For
example, the type well-formedness judgment $\istype A$ and the term well-formedness
judgment $\mtyping M A$ depend on each other. As previously discussed, a general fact
of the interpretation is that $\trm[M]$ in the syntactic level is interpreted to
$\ttrm[\intp M]$.

\begin{figure}
  \centering\small
  \begin{alignat*}{4}
    &\mbox{Interpretation of domain-level types} \span\span\span \\
    &\intp{\istype{\Pi x : A . B}} &\;=\;& \Pi(\intp{\istype A}, \intp{\istype[\Gamma][\Psi, x : A] B}) \\
    &\intp{\istype{\ty}} &\;=\;& \tty' \\
    &\intp{\istype{\trm[M]}} &\;=\;& \ttrm[\intp{\mtyping M \ty}] \\
    &\mbox{Interpretation of domain-level contexts} \span\span\span\\
    &\intp{\isctx\cdot} &\;=\;& \top \\
    &\intp{\isctx{\Psi, x:A}} &\;=\;& \intp{\isctx \Psi}.\intp{\istype A} \\
    &\intp{\isctx{\psi}} &\;=\;& \psi \\
    &\mbox{Interpretation of domain-level substitutions where $\Psi' = \intp{\isctx \Psi}$}
    \span\span\span \\
    &\intp{\mtyping \cdot \cdot} &\;=\;& ! : \El(\Psi') \to \El(\top) \\
    &\intp{\mtyping{\sigma, M}{\Phi, x : A}} &\;=\;& \ctxext{e_1, e_2} : \El(\Psi') \to
    \El(\Phi'.A')
    \tag*{where $A' = \intp{\istype[\Gamma][\Phi]A}$}\\
    \tag*{and $\Phi' = \intp{\isctx \Phi}$} \\
    \tag*{and $e_1 = \intp{\mtyping \sigma \Phi} : \El(\Psi') \to \El(\Phi')$} \\
    \tag*{and $e_2 = \intp{\mtyping{M}{A[\sigma/\hat{\Phi}]}} : \Tm{\Psi'}{A'\{e_1\}}$} \\
    &\intp{\mtyping[\Gamma][\Psi, \overrightarrow{x : A}]{wk_{\hat\Psi}}\Psi}
    &\;=\;& p^k : \El(\intp{\isctx{\Psi, \overrightarrow{x : A}}}) \to \El(\Psi') 
    \tag*{where $k = |\wvec{x{:}A}|$}\\
    &\mbox{Interpretation of domain-level terms where $\Psi' = \intp{\isctx \Psi}$}
    \span\span\span \\
    &\intp{\mtyping{x}{A}}
    &\;=\;& v\{p^k\} : \Tm{\Psi'}{A'\{p^{k+1}\}}
    \tag*{where $\Psi = \Psi_0, x : A, \overrightarrow{y_i : B_i}$} \\
    \tag*{and $|\overrightarrow{y_i : B_i}| = k$} \\
    & & \tag*{$A' = \intp{\istype[\Gamma][\Psi_0] A}$}\\
    &\intp{\mtyping{\lambda x. M}{\Pi x : A. B}}
    &\;=\;& \Lambda(e) : \Tm{\Psi'}{\Pi(A', B')}
    \tag*{where $A' = \intp{\istype A}$ and $B' = \intp{\istype[\Gamma][\Psi, x
        : A]B}$,} \\
    \tag*{and $e = \intp{\mtyping[\Gamma][\Psi, x : A] M B} : \Tm{\Psi'.A'}{B'}$} \\
    &\intp{\mtyping{M\;N}{[N/x]B}}
    &\;=\;& \App(e_1, e_2) : \Tm{\Psi'}{B'\{e_2\}}
    \tag*{where $A' = \intp{\istype A}$ and $B' = \intp{\istype[\Gamma][\Psi, x
        : A]B}$,} \\
    \tag*{and $e_1 = \intp{\mtyping M {\Pi x : A. B}} : \Tm{\Psi'}{\Pi(A', B')}$,}\\
    \tag*{and $e_2 = \intp{\mtyping N A} : \Tm{\Psi'} A$} \\
    &\intp{\mtyping \ttobj \ty} &\;=\;& \tobj : \Tm{\Psi'}{\tty'} \\
    &\intp{\mtyping{\ttarr\;a\;b}{\ty}} &\;=\;& \ttarr({\intp{\mtyping a \ty}},{\intp{\mtyping b
        \ty}}) : \Tm{\Psi'}{\tty'} \\
    &\intp{\mtyping{\tlam\;a\;b\;m}{\trm[(\tarr\;a\;b)]}} 
    &\;=\;& \ttlam (a', b', \App(e'\{p\}, v)) : \Tm{\Psi'}{\ttrm[arr(a',b')]}
    \tag*{where $a' = \intp{\mtyping a \ty} : \Tm{\Psi'}{\tty'}$,} \\
    \tag*{and $b' = \intp{\mtyping b \ty} : \Tm{\Psi'}{\tty'}$,} \\
    \tag*{and $e' = \intp{\mtyping{m}{\trm[a] \to \trm[b]}}: \Tm{\Psi'}{\Pi(\ttrm[a'], \ttrm[b'\{p\}])}$,} \\
    \tag*{and $v : \Tm{\Psi'.\ttrm[a']}{\ttrm[a']\{p\}}$} \\
    &\intp{\mtyping{\tapp\;a\;b\;m\;n}{\trm[b]}}
    &\;=\;& \ttapp(a', b', e_1,e_2)  : \Tm{\Psi'}{\ttrm[b']}
    \tag*{where $a' = \intp{\mtyping a \ty}$ and $b' = \intp{\mtyping b \ty}$,} \\
    \tag*{and $e_1 = \intp{\mtyping m {\trm[(\tarr\;a\;b)]}}$} \\
    \tag*{and $e_2 = \intp{\mtyping n {\trm[a]}}$} \\
    &\intp{\mtyping{\unquote t \sigma}{[\sigma/\hat{\Phi}]A}}
    &\;=\;& \letbox{u}{e_1}{u\{e_2\}} : \Tm{\Psi'}{A'\{e_2\}} \\
    \tag*{where $\Phi' = \intp{\isctx{\Phi}}$ and $A' = \intp{\istype[\Gamma][\Phi] A}$,} \\
    \tag*{and $e_1 = \intp{\typing t {\quot{\judge[\Phi]{A}}}} : \flat(\Tm{\Phi'}{A'})$,} \\
    \tag*{and $e_2 = \intp{\mtyping \sigma \Phi} : \El(\Psi') \to \El(\Phi')$}
  \end{alignat*}
  \caption{Interpretation of the domain-level level}\label{fig:intplf}
\end{figure}

In the interpretation, we proceed by interpreting domain-level types, $\intp{\istype
  A}$. We interpret $\ty$ and $\trm$ to their semantic correspondences. Dependent
function types $\Pi$ is interpreted to semantic $\Pi$  types as defined in the
previous subsection. The interpretation of domain-level contexts $\intp{\isctx \Psi}$
is defined recursively by appending interpreted types to the end of the semantic
context. The empty context $\cdot$ is interpreted to the terminal object $\top$ and
a context variable $\psi$ is interpreted to an interpreted crisp variable in the
interpretation of $\Gamma$, which we will give later in the interpretation of the
computation level. 

Domain-level substitutions are interpreted to substitution morphisms by
$\intp{\mtyping \sigma \Phi}$. The interpretation in this case is very similar to the
simply typed case. We use $!$ for the case of an empty substitution, $\ctxext{e_1,
  e_2}$ for an extended substitution, and iterated first projections $p^k$ for
weakening substitutions.

Last we interpret domain-level terms using $\intp{\mtyping t A}$. In the variable
case, given well-typedness, we know that the domain-level context $\Psi$ must have the
form $\Psi_0, x : A, \overrightarrow{y_i : B_i}$. We first use $v :
\tTm{\Psi_0.A'}{A'\{p\}}$ to extract $x$ from $\Psi_0, x : A$. Then we apply weakening
$p^k$ to it where $k$ is the length of the part of context after $x$, which gives us
$v\{p^k\}$. The abstraction and the application cases are straightforward; they are
interpreted to their semantic terms immediately. Next we interpret the constructors on
the object level. $\ttobj$ and $\ttarr$ are interpreted directly to their semantic
correspondences. 

The $\ttlam$ case is more interesting, because we need to determine how HOAS encoding in the domain language
corresponds in the semantics. According to the rule for $\tlam$ in the previous
subsection, HOAS corresponds to a semantic term in the set
$\tTm{\Phi'.\ttrm[a']}{\ttrm[b'\{p\}]}$. Meanwhile, if we directly interpret $m$ as in the
rule, we obtain a term $e'$ in the set $\tTm{\Psi'}{\Pi(\ttrm[a'],
  \ttrm[b'\{p\}])}$. Therefore, we need to transform $e'$ properly by supplying
$\App(e'\{p\}, v)$. We can examine that this transformation does achieve the goal:
\begin{align*}
  \App(e'\{p\}, v) &: \tTm{\Psi'.\ttrm[a']}{\ttrm[b'\{p\}]\{p\}\{\ctxext{id_{\Psi'.\ttrm[a']}, v}\}} \\
                  &=
                    \tTm{\Psi'.\ttrm[a']}{\ttrm[b'\{p\}]\{p \circ \ctxext{id_{\Psi'.\ttrm[a']},
                    v}\}} \\
                  &= \tTm{\Psi'.\ttrm[a']}{\ttrm[b'\{p\}]}
\end{align*}

The $\ttapp$ case is straightforward. In the unbox case, we first interpret $t$, from
which we obtain a boxed semantic term. We use \kw{let box} to extract from it $u$ in
$\tTm{\Phi'}{A'}$. By applying the interpreted substitution $e_2$ to $u$, we obtain
a term in the expected set $\tTm{\Psi'}{A'\{e_2\}}$.

\subsection{Interpreting the Computation Level}

\begin{figure}
  \centering\small
  \begin{alignat*}{4}
    \mbox{Intepretation of computation types} \span\span\span \\
    \intp{\Istype{\quot{T}}} &\;=\;&& \flat\intp{\judge[\Gamma]T} \\
    \intp{\Istype{(x : \check{\tau_1}) \Rightarrow \tau_2}}
    &\;=\;&& (x :: \intp{{\Istype{\check\tau_1}}}) \to^\flat {\intp{\Istype[\Gamma, x: \ann\tau_1]{\tau_2}}} \\
    \intp{\Istype{\ctx}} &\;=\;&& \Ctx \\
    \mbox{Interpretation of computation contexts} \span\span\span\\
    \intp{\cdot} &\;=\;&& \cdot \\
    \intp{\Gamma, x :\check{\tau}} &\;=\;&& \intp{\Gamma}, x :: \intp{\Istype{\ann\tau}}  \\
    \mbox{Interpretation of contextual objects} \span\span\span\\
    \intp{\typing{(\judge[\hat{\Psi}]{M})}{(\judge[\Psi]A)}}
    &\;=\;&& \intp{\mtyping M A} \\
    \intp{\typing{(\judge[\hat{\Psi}]{M})}{(\vjudge\Psi A)}}
    &\;=\;&& \intp{\mtyping M A} \\
    \mbox{Interpretation of contextual types} \span\span\span\\
    \intp{\judge{(\judge[\Psi]A)}}
    &\;=\;&& \Tm{\intp{\isctx{\Psi}}}{\intp{\istype A}} \\
    \intp{\judge{(\vjudge\Psi A)}}
    &\;=\;&& \Tmv{\intp{\isctx{\Psi}}}{\intp{\istype A}}
    \tag*{$\CTm_v$ is $\CTm$ but equivalent to $v\{p^k\}$}\\
    \mbox{Interpretation of computation terms} \span\span\span\\
    \intp{\typing{\quot{C}}{\quot{T}}}
    &\;=\;&& \boxit{\intp{\typing C T}} \\
    \intp{\typing{t\;s}{[s/y]\tau_2}}
    &\;=\;&& \intp{\typing{t}{(y : \check{\tau_1}) \Rightarrow \tau_2}}\;
    \intp{\typing{s}{\check{\tau_1}}} \\
    \intp{\typing{\deffun x t}{(x : \check{\tau_1}) \Rightarrow \tau_2}}
    &\;=\;&& \mlam x :: \flat\intp{\Istype{\ann\tau_1}} .{\intp{\typing[\Gamma, x :
      \check{\tau_1}]{t}{\tau_2}}} \\
    \intp{\typing x{\check\tau}} &\;=\;&& x
  \end{alignat*}
  \caption{Interpretation of the computation level without recursors}\label{fig:intpcomp}
\end{figure}

\begin{figure}[hbt]
  \centering\small
  The semantic recursor for $\ttrm$:
  \begin{mathpar}
    \inferrule*
    {\mmtyping[\Gamma][\cdot] \Psi \Ctx \\
      \mmtyping[\Gamma][\cdot] a {\flat(\Tm\top \tty)} \\
      \mmtyping[\Gamma][\cdot] y {\flat(\Tm\Psi{\ttrm[\lift(\Psi, a)]})}}
    {\Istype{R(\Psi, a, y)}}

    \inferrule*
    {\mmtyping[\Gamma][\cdot] \Psi \Ctx \\
      \mmtyping[\Gamma][\cdot] a {\flat(\Tm\top \tty)} \\
      \mmtyping[\Gamma][\cdot] t {\flat(\Tmv\Psi{\ttrm[\lift(\Psi, a)]})}}
    {\mmtyping[\Gamma][\cdot]{B_v(\Psi, a, t)}{R(\Psi, a, t)}}

    \inferrule*
    {\mmtyping[\Gamma][\cdot] \Psi \Ctx \\
      \mmtyping[\Gamma][\cdot] a {\flat(\Tm\top \tty)} \\
      \mmtyping[\Gamma][\cdot] b {\flat(\Tm\top \tty)} \\
      \mmtyping[\Gamma][\cdot] m {\flat(\Tm{\Psi.\ttrm[\lift(\Psi, a)]}{\ttrm[\lift(\Psi.\ttrm[\lift(\Psi, a)], b)]})} \\
      \mmtyping[\Gamma][\cdot]{f_m}{R(\Psi.\ttrm[\lift(\Psi, a)], b, m)}
    }
    {\mmtyping[\Gamma][\cdot]{B_{\tlam}(\Psi, m, f_m)}{R(\Phi, \tarr'(a, b), \tlam'(\Psi, m))}}
    
    \inferrule*
    {\mmtyping[\Gamma][\cdot] \Psi \Ctx \\
      \mmtyping[\Gamma][\cdot] a {\flat(\Tm\top \tty)} \\
      \mmtyping[\Gamma][\cdot] b {\flat(\Tm\top \tty)} \\\\
      \mmtyping[\Gamma][\cdot] m {\flat(\Tm\Psi{\ttrm[\tarr(\lift(\Psi, a), \lift(\Psi, b))]})} \\
      \mmtyping[\Gamma][\cdot] n {\flat(\Tm\Psi{\ttrm[\lift(\Psi, a)]})} \\
      \mmtyping[\Gamma][\cdot]{f_m}{R(\Psi, \tarr'(a, b), m)} \\
      \mmtyping[\Gamma][\cdot]{f_n}{R(\Psi, a, n)} 
    }
    {\mmtyping[\Gamma][\cdot]{B_{\tapp}(\Psi, m, n, f_m, f_n)}{R(\Psi, b, \tapp'(\Psi, m, n))}}

    \inferrule*
    { }
    {\Gamma \pipe \cdot \vdash{rec_{\ttrm}(B_v, B_{\tlam}, B_{\tapp})} : (\Psi : \Ctx) \to (a :
      \flat(\Tm\top \tty)) \to (y : \flat(\Tm\Psi{\ttrm[\lift(\Psi, a)]})) \to R(\Psi, a, y)}
  \end{mathpar}

  \vspace{1em}
  Equations:
  \begin{align*}
    &rec_{\ttrm}(B_v, B_{\tlam}, B_{\tapp}, \Psi, a, x) = B_v(\Psi, a, x) \tag*{where $x : \flat(\Tmv\Psi{\ttrm[\lift(\Psi, a)]})$}\\
    &rec_{\ttrm}(B_v, B_{\tlam}, B_{\tapp}, \Psi, \tarr'(a, b), \tlam'(\Psi,
      m)) \\
    = &B_{\tlam}(\Psi, f, rec_{\ttrm}(B_v, B_{\tlam}, B_{\tapp},
        \Psi.\ttrm[\lift(\Psi, a)], b, m)) \\
    &rec_{\ttrm}(B_v, B_{\tlam}, B_{\tapp}, \Psi, b, \tapp'(\Psi, m, n)) \\
    = &B_{\tlam}(\Psi, f, m, n, rec_{\ttrm}(B_v, B_{\tlam}, B_{\tapp}, \Psi,
        \tarr'(a, b), m),
        rec_{\ttrm}(B_v, B_{\tlam}, B_{\tapp}, \Psi, a, n))
  \end{align*}

  \caption{Semantic recursor}\label{fig:recursor}
\end{figure}
\begin{figure}[hbt]
  \centering\small
  \begin{alignat*}{4}
    \mbox{The interpretation of the recursor for $\ttrm$ where $\mathcal{I} = (\psi :
      \ctx) \Rightarrow (z : \quot{\judge[]{\ty}}) \Rightarrow (y :
      \quot{\judge[\psi]{\trm[\unquote{z}{\cdot}]}}) \Rightarrow \tau$}\span\span\span
    \\
    \intp{\typing{\induct{\mathcal{I}}{(b_v \pipe b_{\ttlam} \pipe
          b_{\ttapp})}{\Psi}{t}\;t'}{\tau[\Psi,t,t'/\psi, z, y]}}
    &\;=\;&& rec_{\ttrm}(e_v, e_{\tlam}, e_{\tapp}, e_{\Psi}, e_t, e_{t'})
    \tag*{where $e_\Psi = \intp{\isctx \Psi}$} \\
    \tag*{and $e_t = \intp{\typing t {\quot{\judge[]\ty}}}$} \\
    \tag*{and $e_{t'} = \intp{\typing{t'}{\quot{\judge[\Psi]\trm[\unquote{t}{\cdot}]}}}$} \\
    \tag*{and $e_v = \intp{\typing{b_v}{\mathcal{I}}}$} \\
    \tag*{and $e_{\tlam} = \intp{\typing{b_{\ttlam}}{\mathcal{I}}}$} \\
    \tag*{and $e_{\tapp} = \intp{\typing{b_{\ttapp}}{\mathcal{I}}}$} \\
  \end{alignat*}
  \begin{alignat*}{4}
    \mbox{Interpretation of branches} \span\span\span \\
    \intp{\typing{b_v}{\mathcal{I}}} &\;=\;&& \mlam\Psi\;a\;t.e
    \tag*{where $e = \intp{\typing[\Gamma, \Psi : \ctx, a : \quot{\judge[]\ty}, t :
    \quot{\vjudge{\Psi}{\trm[\unquote a {\cdot}]}}]{t_v}{\tau[a, t/ z, y]}}$}
    \\
    \intp{\typing{b_{\ttlam}}{\mathcal{I}}} &\;=\;&& \mlam \Psi\;a\;b\;m\;f_m.e
    \tag*{where $e = \intp{\text{the premise judgment}}$} \\
    \intp{\typing{b_{\ttapp}}{\mathcal{I}}} &\;=\;&& \mlam
    \Psi\;a\;b\;m\;n\;f_m\;f_n.e
    \tag*{where $e = \intp{\text{the premise judgment}}$}
  \end{alignat*}
  \caption{Interpretation of recursor for $\ttrm$}\label{fig:intprec}
\end{figure}

In this section, we discuss the interpretation of the computation level of \cocon. The
interpretation of the computation level is simpler than the one of the domain level,
because we work in a presheaf category which possesses a CwF structure with $\Pi$
types. The interpretation only needs to relate the corresponding parts. 

The interpretation functions without the recursors are shown in
\Cref{fig:intpcomp}. They are very similar to the simply typed case. In the
interpretation of computation types, $\intp{\Istype{\ann\tau}}$, we surround boxed
contextual types with a $\flat$ modality, and the contextual types are interpreted
using $\intp{\judge{(\judge[\Psi]A)}}$, resulting in some $\CTm$.
Computation-level
functions are directly translated to crisp functions in the model, similar to the
simply typed case. 
This is because we want to maintain the invariant
where computation-level variables are all crisp and in the model all crisp variables
live in $\flat$.
At last, we simply map $\ctx$ to $\Ctx$ which is a type representing the domain-level
contexts.

The interpretation of computation contexts, $\intp{\Gamma}$, simply iteratively
interprets all types in it. Note that a computation-level context $\Gamma$ is
interpreted to a global or crisp context in our model. That is why we do not wrap all
interpreted types with $\flat$ and seemingly mismatches with the parameter types in
function types. We will resolve this problem when interpreting computation-level
terms.

The interpretations of contextual objects and contextual terms are immediately reduced
to the interpretations of domain-level types and terms, which we have discussed in the
previous subsection. Since $\vjudge\Psi A$ denotes a variable of type $A$ in domain
context $\Psi$, the interpretation as a $\tTm\Phi A$ is restricted such that it is
semantically also a variable lookup $v\{p^k\}$. 
restricted semantic term in the set $\tTm\Phi A$ in the form of $v\{p^k\}$. This is indeed
the case by looking at the interpretation of the variable case of the domain level.



The interpretation of the computation-level terms is straightforward. Boxed contextual
objects are simply interpreted as boxed domain-level terms in the model. Since we
interpret to computation-level functions to crisp functions, we use crisp applications
and $\mlam$ abstractions respectively to interpret computation-level applications and
abstractions. Notice that in the application case, due to the soundness theorem we are
about to show, the interpretation of $s$, $\intp{\typing{s}{\ann\tau_1}}$, is indeed
closed, and thus the application is valid. At last, variables are simply interpreted
to those in the semantic context. 

\subsection{Interpreting Recursors}

Compared to the previous standard interpretations, the recursors are more interesting
to consider. The recursor of $\ty$ appears to be very typical because it is simply a
algebraic data type, which can already be modeled conventionally using the initial
algebra of some polynomial functor. Therefore, we omit the concrete formulation here
in favor of conciseness and only focus on the recursor of $\trm$.

In order to formulate the semantic recursor of $\ttrm$, we define the following three auxiliary definitions:
\begin{mathpar}
  \inferrule*
  {\mmtyping[\Gamma][\cdot]{A}{\tTy\top} \\ \mmtyping[\Gamma][\cdot]{\Psi}{\Ctx} \\ \mmtyping[\Gamma][\cdot]{t}{\flat(\tTm\top A)}}
  {\mmtyping[\Gamma][\cdot]{\lift(\Psi, t) := \letbox{x'}t{x'\{!\}}}{\tTm\Psi{A\{!\}}}}

  \inferrule*
  {\mmtyping[\Gamma][\cdot] a {\flat(\Tm\top \tty)} \\
    \mmtyping[\Gamma][\cdot] b {\flat(\Tm\top \tty)}}
  {\mmtyping[\Gamma][\cdot]{\tarr'(a, b) :=
      \letbox{a'}{a}{\letbox{b'}{b}{\boxit{\tarr(a', b')}}}}
  {\flat(\Tm\top \tty)}}
  
  \inferrule*
  {\mmtyping[\Gamma][\cdot]{\Psi}{\Ctx} \\ \mmtyping[\Gamma][\cdot]{a}{\tTm\Psi{\tty'}} \\
    \mmtyping[\Gamma][\cdot]{b}{\tTm\Psi{\tty'}} \\
    \mmtyping[\Gamma][\cdot]{m}{\flat(\tTm{\Psi.\ttrm[a]}{\ttrm[b\{p\}]})}}
  {\mmtyping[\Gamma][\cdot]{\tlam'(\Psi, m) := \letbox{m'}m{\boxit{\tlam(a, b, m')}}}{\flat(\tTm\Psi{\ttrm[\tarr(a, b)]})}}
\end{mathpar}
\begin{mathpar}
  \inferrule*
  {\mmtyping[\Gamma][\cdot]{\Psi}{\Ctx} \\ \mmtyping[\Gamma][\cdot]{a}{\tTm\Psi{\tty'}} \\
    \mmtyping[\Gamma][\cdot]{b}{\tTm\Psi{\tty'}} \\
    \mmtyping[\Gamma][\cdot]{m}{\flat(\tTm\Psi{\ttrm[arr(a, b)]})} \\
    \mmtyping[\Gamma][\cdot]{n}{\flat(\tTm\Psi{\ttrm[a]})}}
  {\mmtyping[\Gamma][\cdot]{\tapp'(\Psi, m, n) := \letbox{m'}m{\letbox{n'}n{\boxit{\tapp(a, b, m', n')}}}}{\flat(\tTm\Psi{\ttrm[b]})}}
\end{mathpar}
These helpers are defined to ease the operations related to the $\flat$ modality. For
example, the $\lift$ function transforms a term $t$ of type $\flat(\tTm\top A)$ to
$\tTm\Psi{A\{!\}}$ for some domain-level context $\Psi$. $\tlam'$ takes a boxed HOAS
representation, $m$, and return a boxed $\ttrm$ constructed by $\tlam$. Similarly, $\tapp'$
takes two boxed $\ttrm$ and return a boxed $\ttrm$ constructed by $\tapp$. We need
these helpers in order to reduce the clusters in the formulation of the semantic recursor of
$\trm$ and the equations, which is presented in \Cref{fig:recursor}.

Since the recursion happens in the computation level, we require the local context to
be empty, so we only handle closed domain-level types and terms.  In the
variable case, we require the semantic term $t$ to be $\CTm{_v}$, so that it indeed represents
 a variable in $\Psi$. In the $\tlam$ case, the recursion involves a HOAS encoding of an
object-level term. This corresponds to a domain-level term with an augmented context
$\Psi.\ttrm[\lift(\Psi, a)]$. The $\tapp$ case is straightforward since it just goes
down to the subterms recursively. 

Given the semantic recursor, we can straightforwardly interpret the syntactical
recursor, shown in \Cref{fig:intprec}.

Given the recursor, we can interpret the syntactical recursor quite straightforwardly
to the semantic recursor.
Following the pattern from the simply typed case, we interpret branches to
crisp functions and the bodies are recursively interpreted.

That concludes our interpretations. We formulate the soundness properties of the
interpretations as below, which are proved via a mutual induction. 
\begin{theorem}[Soundness]
  The following are true.
  \begin{enumerate}
  \item If $\isctx{\Phi}$, then $\typing[\intp{\Gamma} \pipe \cdot]{\intp{\isctx\Phi}}{\Ctx}$.
  \item If $\istype{A}$, then $\typing[\intp{\Gamma} \pipe \cdot]{\intp{\istype A}}{\tTy{\intp{\isctx\Psi}}}$.
  \item If $\mtyping M A$, then
    $\typing[\intp{\Gamma} \pipe \cdot] {\intp{\mtyping M A}}{\tTm{\intp{\isctx\Psi}}{\intp{\istype A}}}$.
  \item If $\mtyping \sigma \Phi$, then
    $\typing[\intp{\Gamma} \pipe \cdot]{\intp{\mtyping \sigma
        \Phi}}{\El(\intp{\isctx{\Psi}}) \to \El(\intp{\isctx{\Phi}})}$.
  \item If $\typing C T$, then $\typing[\intp{\Gamma} \pipe \cdot]{\intp{\typing C
        T}}{\intp{\Gamma \vdash T}}$.
  \item If $\Istype{\ann\tau}$, then $\Istype[\sem{\Gamma} \pipe
    \cdot]{\intp{\Istype{\ann\tau}}}$.
  \item If $\typing t {\check\tau}$, then $\typing[\intp{\Gamma} \pipe
    \cdot]{\intp{\typing t {\check\tau}}}{\intp{\Istype{\check\tau}}}$.
  \item If $\istype{A \equiv A'}$, then $\typing[\intp{\Gamma} \pipe
    \cdot]{\intp{\istype A} = \intp{\istype{A'}}}{\tTy{\intp{\isctx\Psi}}}$.
  \item
    If $\Gamma\semi \Psi \vdash M \equiv N \colon A$ then
    $\sem{\Gamma}\pipe \cdot \vdash \sem{\Gamma \semi \Psi \vdash M : A} = \sem{\Gamma
      \semi \Psi \vdash N : A} \colon {\tTm{\intp{\isctx\Psi}}{\intp{\istype A}}}$.
  \item
    If
    $\Gamma; \Psi \vdash \sigma \equiv \sigma' \colon \Phi$ then
    $\sem{\Gamma} \pipe \cdot
     \vdash
     \sem{\Gamma; \Psi \vdash \sigma : \Phi} = \sem{\Gamma; \Psi \vdash \sigma' : \Phi}
     \colon {\El(\intp{\isctx{\Psi}}) \to \El(\intp{\isctx{\Phi}})}$.
  \item If $\Istype{\ann{\tau_1} \equiv  \ann{\tau_2}}$, then $\Istype[\sem{\Gamma} \pipe
    \cdot]{\intp{\Istype{\ann{\tau_1}}} = \intp{\Istype{\ann{\tau_2}}}}$.
  \item
    If
    $\Gamma \vdash t_1 \equiv t_2 \colon \tau$,
    $\sem{\Gamma} \pipe \cdot \vdash \sem{\Gamma \vdash t_1 : \ann\tau} = \sem{\Gamma \vdash t_2 : \ann\tau} \colon \sem{\Istype{\ann\tau}}$.
  \end{enumerate}
\end{theorem}


\section{Connection with Fitch-Style Type Theories}\label{sec:fitch}

\citet{DBLP:journals/mscs/BirkedalCMMPS20,DBLP:journals/pacmpl/GratzerSB19} discussed
two Fitch-style dependent modal type theories. Compared to the dual-context-style we
presented in previous sections, Fitch-style systems differ in that they use only one context to keep
track of all variables. Instead, Fitch-style systems use a ``locking'' mechanism to
prevent variable lookups from continuing. In this section, we establish a relation
between a fragment of \cocon without recursion on HOAS and a semantic framework discussed in
\citet{DBLP:journals/mscs/BirkedalCMMPS20}, dependent right adjoints. We show this by
showing an embedding of \cocon into the dependent intuitionistic K shown in
\citet{DBLP:journals/mscs/BirkedalCMMPS20}, which has the soundness and completeness
properties with respect to dependent right adjoints. Thus we can establish that the
fragment of \cocon without recursion on HOAS
can be interpreted by any system with dependent right adjoints.

\subsection{Fitch-style Modal Type Theories}

Fitch-style modal type theories are more intuitive than dual-context-style modal type
theories in a sense that Fitch style only handles one context. As a consequence, valid
and true assumptions in the contexts are mixed together. The introduction and the
elimination rules thus must tell these different kinds of assumption apart. In
Fitch-style systems, the introduction rule for necessity or box introduces a lock to
the top of the context. This lock ``blocks'' the context to its left. Different
flavors of Fitch-style systems are distinguished by their elimination rules. For
dependent intuitionistic K, the rules are
\begin{mathpar}
  \inferrule*
  {\typing[\Gamma, \thelock]m T}
  {\typing{\boxit{m}}{\square T}}

  \inferrule*
  {\typing m{\square T} \\ \thelock \notin \Gamma'}
  {\typing[\Gamma, \thelock, \Gamma']{\tunbox{m}}{T}}
\end{mathpar}
In both rules, \thelock\xspace symbol prevents variable lookups from going beyond its left. That
is, a term can only refer to variables to the right of the rightmost \thelock.  The
rules can be understood from the classical Kripke's semantics: \kw{box} accesses the
next world, and thus the \thelock\xspace symbol locks the assumptions in all previous
worlds (to its left); while \kw{unbox} allows us to only travel back to the immediate previous
world, so \thelock\xspace must not exist in $\Gamma'$ in the elimination rule. \cocon is much closer to
dependent intuitionistic K than to the dual-context models we discussed in the previous
sections.  Nonetheless, one fundamental problem of dependent intuitionistic K is that it does
not support recursion on HOAS structures like $\trm$ in \Cref{sec:depty} as \cocon
does. Therefore, in this section, we will only discuss the interpretation from the
fragment of \cocon without recursors to dependent intuitionistic K in order to discuss their
connection explicitly. Whether dependent intuitionistic K can be extended to support
recursion on HOAS is an interesting topic for future investigation.

\subsection{Dependent Right Adjoints}

Dependent right adjoints are a special structure added on top of a category with
families. Essentially it is a CwF equipped with a functor $L$, denoting the
\thelock\xspace symbol, and a family of types $R$, denoting the modality. Dependent
right adjoints are used to capture the nature of comonadic modality in a categorical
language.
\begin{definition}
  A category with families with a dependent right adjoint $\mathcal{C}$ is a
  category with families with the following extra data:
  \begin{enumerate}
  \item an endofunctor $L : \mathcal{C} \to \mathcal{C}$,
  \item a family $R_\Gamma(A) \in Ty(\Gamma)$ for each $\Gamma \in \mathcal{C}$ and $A
    \in Ty(L(\Gamma))$. 
  \end{enumerate}
  The following axioms hold:
  \begin{enumerate}
  \item $R_\Gamma(A) \{ \sigma\} = R_\Delta(A \{L(\sigma)\}) \in Ty(\Delta)$ for
    $\sigma : \Delta \to \Gamma$,
  \item the following isomorphism exists for $\Gamma \in \mathcal{C}$ and $A \in Ty(L(\Gamma))$:
    \begin{align*}
      Tm(L(\Gamma), A) \simeq Tm(\Gamma, R_\Gamma(A))
    \end{align*}
    with the effect from left to right as
    $\overrightarrow{M} \in Tm(\Gamma, R_\Gamma(A))$ for $M \in Tm(L(\Gamma), A)$ and
    the other effect $\overleftarrow{N} \in Tm(L(\Gamma), A)$ for
    $N \in Tm(\Gamma, R_\Gamma(A))$.
  \item $\overrightarrow{M}\{\sigma\} = \overrightarrow{M\{L(\sigma)\}} \in Tm(\Delta,
    R_\Delta(A\{\sigma\}))$ for $M
    \in Tm(L(\Gamma), A)$ and $\sigma \in \Delta \to \Gamma$.
  \end{enumerate}
\end{definition}

\citet{DBLP:journals/mscs/BirkedalCMMPS20} shows that dependent intuitionistic K
can be soundly interpreted into a CwF with a dependent right adjoint and
the type theory itself forms a term model. Thus it suffices to interpret \cocon
to their type theory to show that \cocon can be interpreted using the structure
of a dependent right adjoint. We give the interpretation in the next section.

\begin{figure}
  \centering\small
  \begin{alignat*}{4}
    &\mbox{Interpretation of domain-level types} \span\span\span \\
    &\intp{\istype{\Pi x : A . B}} &\;=\;& \lamb u. \Pi x : \intp{\istype A}\
    u. \intp{\istype[\Gamma][\Psi, x : A] B}\ u \\
    &\mbox{Interpretation of domain-level contexts} \span\span\span\\
    &\intp{\isctx\cdot} &\;=\;& \top \\
    &\intp{\isctx{\Psi, x:A}} &\;=\;& \intp{\isctx \Psi}. \intp{\istype A} \\
    &\intp{\isctx{\psi}} &\;=\;& \tunbox\psi \\
    &\mbox{Interpretation of domain-level substitutions where $\Psi' = \intp{\isctx \Psi}$}
    \span\span\span \\
    &\intp{\mtyping \cdot \cdot} &\;=\;& \lamb u. () \\
    &\intp{\mtyping{\sigma, M}{\Phi, x : A}} &\;=\;& \lamb u. (e_1~u, e_2~u)
    \tag*{where $e_1 = \intp{\mtyping \sigma \Phi}$ and $e_2 = \intp{\mtyping{M}{A[\sigma/\hat{\Phi}]}}$} \\
    &\intp{\mtyping[\Gamma][\Psi, \overrightarrow{x : A}]{wk_{\hat\Psi}}\Psi}
    &\;=\;& \pi_1^k \tag*{where $k = |\wvec{x{:}A}|$}\\
    &\mbox{Interpretation of domain-level terms where $\Psi' = \intp{\isctx \Psi}$}
    \span\span\span \\
    &\intp{\mtyping{x}{A}}
    &\;=\;& \lamb u. \pi_2(\pi_1^k~u)
    \tag*{where $\Psi = \Psi_0, x : A, \overrightarrow{y_i : B_i}$
      and $|\overrightarrow{y_i : B_i}| = k$} \\
    &\intp{\mtyping{\lambda x. M}{\Pi x : A. B}}
    &\;=\;& \lamb u~x. e~(u, x)
    \tag*{where $e = \intp{\mtyping[\Gamma][\Psi, x : A] M B}$} \\
    &\intp{\mtyping{M\;N}{[N/x]B}}
    &\;=\;& \lamb u. (e_1~u)~(e_2~u)
    \tag*{where $e_1 = \intp{\mtyping M {\Pi x : A. B}}$ and $e_2 = \intp{\mtyping N
        A}$} \\
    &\intp{\mtyping{\unquote t \sigma}{[\sigma/\hat{\Phi}]A}}
    &\;=\;&  \lamb u. (\tunbox{e_1})~(e_2~u)
    \tag*{where $e_1 = \intp{\typing t {\quot{\judge[\Phi]{A}}}}$ and $e_2 =
      \intp{\mtyping \sigma \Phi}$} \\
    &\mbox{Interpretation of contextual objects} \span\span\span\\
    &\intp{\typing{(\judge[\hat{\Psi}]{M})}{(\judge[\Psi]A)}}
    &\;=\;& \intp{\mtyping M A} \\
    &\intp{\typing{(\judge[\hat{\Psi}]{M})}{(\vjudge\Psi A)}}
    &\;=\;& \intp{\mtyping M A} \\
    &\mbox{Interpretation of contextual types} \span\span\span\\
    &\intp{\judge{(\judge[\Psi]A)}}
    &\;=\;& \Pi u : {\intp{\isctx{\Psi}}}.{\intp{\istype A}} \\
    &\intp{\judge{(\vjudge\Psi A)}}
    &\;=\;& \Pi u : {\intp{\isctx{\Psi}}}.{\intp{\istype A}} \\
    &\mbox{Intepretation of computation types} \span\span\span \\
    &\intp{\Istype{\quot{T}}} &\;=\;& \square\intp{\judge T} \\
    &\intp{\Istype{(x : \ann\tau_1) \Rightarrow \tau_2}}
    &\;=\;& \Pi x : \intp{\Istype{\ann\tau_1}}. \intp{\Istype{\tau_2}} \\
    &\intp{\Istype{\ctx}} &\;=\;& \square\Ctx \\
    &\mbox{Interpretation of computation contexts} \span\span\span\\
    &\intp{\cdot} &\;=\;& \cdot \\
    &\intp{\Gamma, x :\check{\tau}} &\;=\;& \intp{\Gamma}, x : \intp{\Istype{\ann\tau}}  \\
    &\mbox{Interpretation of computation terms} \span\span\span\\
    &\intp{\typing{\quot{C}}{\quot{T}}}
    &\;=\;& \boxit{\intp{\typing C T}} \\
    &\intp{\typing{t_1\;t_2}{[t_2/x]\tau}}
    &\;=\;& \intp{\typing{t_1}{(x : {\ann\tau_2}) \Rightarrow
        \tau}}~\intp{\typing{t_2}{{\ann\tau_2}}} \\ 
    &\intp{\typing{\deffun x t}{(x : \ann{\tau_1}) \Rightarrow \tau_2}}
    &\;=\;& \lamb x : \intp{\ann{\tau_1}} .{\intp{\typing[\Gamma, x :
        \ann{\tau_1}]{t}{\tau_2}}} \\
    &\intp{\typing x {\check\tau}} &\;=\;& x
  \end{alignat*}
  \caption{Interpretation to the Fitch-style system}\label{fig:intpfitch}
\end{figure}

\subsection{Interpreting \cocon}

The interpretation is mostly straightforward. We can easily show that the
interpretation is also sound:
\begin{theorem}[Soundness]
  The following are true.
  \begin{enumerate}
  \item If $\isctx{\Phi}$, then $\typing[\intp{\Gamma}]{\intp{\isctx\Phi}}{\Ctx}$.
  \item If $\istype{A}$, then $\typing[\intp{\Gamma}]{\intp{\istype A}}{\Pi u :
      \intp{\isctx\Psi}. \Type}$.
  \item If $\mtyping M A$, then
    $\typing[\intp{\Gamma}] {\intp{\mtyping M A}}{\Pi u
      : \intp{\isctx\Psi}. \intp{\istype A}\ u}$.
  \item If $\mtyping \sigma \Phi$, then
    $\typing[\intp{\Gamma}]{\intp{\mtyping \sigma
        \Phi}}{\Pi u : \intp{\isctx\Psi}. \intp{\isctx{\Phi}}}$.
  \item If $\typing C T$, then $\typing[\intp{\Gamma}]{\intp{\typing C
        T}}{\intp{\Gamma \vdash T}}$.
  \item If $\Istype{\ann\tau}$, then $\typing[\sem{\Gamma}]{\intp{\Istype{\ann\tau}}}{\Type}$.
  \item If $\typing t {\check\tau}$, then $\typing[\intp{\Gamma}]{\intp{\typing t {\check\tau}}}{\intp{\Istype{\check\tau}}}$.
  \end{enumerate}
\end{theorem}

There are a number of differences to
highlight:
\begin{enumerate}
\item In \cocon, the domain level sees contextual variables on the computation
  level. This implies that contextual variables must live in $\square$. 
\item The domain level and the computation level in \cocon have different
  syntax. In the interpretation, we need to merge them, e.g. two dependent function
  spaces are merged into the same syntax. We can still distinguish them by looking at
  the level they live in. 
\end{enumerate}

This idea leads to an interpretation shown in \Cref{fig:intpfitch}. As we can see, the
interpretation is very straightforward, showing that the Fitch-style system is
compatible with \cocon in many aspects. 

In the interpretation of the domain-level types, we only show the case for $\Pi$
types. If we have corresponding base types in \cocon and dependent intuitionistic K,
we can relate them via the interpretation. For example, if we have $\tty$ in dependent
intuitionistic K as well, then we will have a base case in the
interpretation. Nonetheless, we can still work on other parts of the interpretations.

In dependent intuitionistic K, we assume a universe $\Ctx$, which is used to represent
domain-level contexts. There are two types, $\top$ to represent the empty context and
$-.-$ for an extended context. That is, a domain-level context are managed as a
list-like structure. To construct a domain-level context in dependent
intuitionistic K, we use $()$ to construct an empty context and $-,-$ to extend an
existing context. Given $\Phi.A$, we can get the precedent $\Phi$ by applying $\pi_1$
and get the domain-level term of type $A$ by applying $\pi_2$. This is sufficient for
us to perform operations related domain-level contexts in dependent intuitionistic
K. The interpretation of domain-level contexts is quite straightforward, except that
when we encounter a context variable $\psi$, we unbox it in the interpretation.  This
is because contextual variables always have type $\square \Ctx$ in the model, as to be shown in the
interpretation of computation-level types.

The interpretations of domain-level substitutions and terms are also very
straightforward. Notice that the unbox case in the interpretation of terms becomes
much more direct. Since the interpretation of $t$ is some boxed function type, we
simply apply it to the interpretation of the substitution $\sigma$ after
unboxing. Notice that there we do not have to worry about idempotency of $\square$ as
in previous sections, as unbox in \cocon is already naturally modeled by unbox in
dependent intuitionistic K.

The interpretation of contextual objects directly forward to the interpretation of
domain-level terms. Contextual types are interpreted to dependent functions.

For computation-level types, we interpret boxed contextual types to boxed types and
function types in \cocon to function types in dependent intuitionistic K as typically
done. Notice that we interpret function types on both the domain level and the
computation level to functions in dependent intuitionistic K, so that we have a unified
syntax. It is worth mentioning that $\ctx$ is interpreted to $\square \Ctx$, because
context variables are always global in \cocon. That is why we have an additional \kw{unbox}
when interpreting contextual variables.

The interpretations of computation-level contexts and terms are immediate.

Based on this interpretation, the domain and the computation levels reside separately in
two ``zones'' in dependent intuitionistic K. These two ``zones'' can be distinguished
by checking whether a \thelock\xspace exists in the context. If there is, then the current
term is on the computation level, and otherwise it is on the domain level. A \thelock\xspace
is added when a \kw{box} is encountered. This corresponds to getting into the domain
level from the computation one. Therefore, dependent intuitionistic K and \cocon do
seem to correspond nicely (except that the former does not have recursors for HOAS). 

\section{Conclusion}

We have given a rational reconstruction of contextual type theory in presheaf models of
higher-order abstract syntax.
This provides a semantical way of understanding the invariants of contextual types
independently of the algorithmic details of type checking.
At the same time, we identify the contextual modal type theory, \cocon, which is known to be normalizing, as a syntax for presheaf
models of HOAS.
By accounting for the Yoneda embedding with a universe \'a la Tarski, we obtain a manageable
way of constructing contextual types in the model, especially in the dependent
case.
Presheaves over models of dependent types have been used in the context of two-level type 
theories for homotopy type theory~\cite{capriotti:thesis, ann-cap-kra:two-level}.
Clarifying the precise relationship to this line of research is an interesting direction
that will however require further work.
  
In future work, one may consider using the model as a way of compiling contextual types,
by implementing the semantics.
In another direction, it may be interesting to apply the syntax of contextual types to other presheaf categories.
We also hope that the model will help to guide the further development of \cocon.

\section*{Acknowledgement}

This work was funded by the Natural Sciences and Engineering Research Council of
Canada (grant number 206263), Fonds de recherche du Qu\'ebec - Nature et Technologies
(grant number 253521), and Postgraduate Scholarship - Doctoral by the Natural Sciences
and Engineering Research Council of Canada awarded to the first author.

\bibliographystyle{ACM-Reference-Format}


\end{document}